\documentclass[lettersize,onecolumn]{IEEEtran}
\usepackage{amsmath,amsfonts,amsthm}
\usepackage{amssymb}
\usepackage{algorithmic}
\usepackage{algorithm}
\usepackage{array}
\usepackage[caption=false,font=normalsize,labelfont=sf,textfont=sf]{subfig}
\usepackage{textcomp}
\usepackage{stfloats}
\usepackage{url}
\usepackage{verbatim}
\usepackage{graphicx}

\usepackage[english]{babel}

\usepackage{bm}
\usepackage{xcolor}
\usepackage[numbers]{natbib}
\usepackage[normalem]{ulem}
\usepackage{bbm}
\usepackage{tikz}
\usepackage{soul}
\usepackage{hyperref}
\usepackage{cleveref}
\usepackage{nicefrac}

\newcommand{\Ex}{\expec}
\newcommand{\calV}{\Vcal}
\newcommand{\calW}{\Wcal}
\newcommand{\remainder}{\varrho}

\numberwithin{equation}{section} 
\newcommand{\pow}[1]{^{(#1)}}
\newcommand{\lp}{\left(}
\newcommand{\rp}{\right)}
\newcommand{\cR}{\mathcal{R}}
\newcommand{\cQ}{\mathcal{Q}}

\newcommand{\lb}{\left[}
\newcommand{\rb}{\right]}

\newcommand{\Var}{\mathrm{Var}}
\newcommand{\Cov}{\mathrm{Cov}} 
\newcommand{\bsq}{\vrule height .9ex width .8ex depth -.1ex}

\newcommand{\JGF}{J^{GF}}

\newenvironment{proofof}[1]{%
	\noindent\begin{proof}[{\sc Proof of #1}]%
	}{%
	\end{proof}%
}

\newcommand{\Vbb}{\mathbb{V}}
\newcommand{\Wbb}{\mathbb{W}}

\newcommand{\Fcal}{\mathcal{F}}

\newcommand{\Ocal}{\mathcal{O}}
\newcommand{\Pcal}{\mathcal{P}}

\newcommand{\Vcal}{\mathcal{V}}
\newcommand{\Wcal}{\mathcal{W}}

\usepackage{dutchcal}

\DeclareMathAlphabet{\mathdutchcal}{U}{dutchcal}{m}{n}

\newcommand{\ocal}{\mathbcal{o}}

\newcommand{\hx}{\hat{x}}

\newcommand\numberthis{\addtocounter{equation}{1}\tag{\theequation}}

\newcommand{\whiteqed}{\hfill$\square$\par\bigskip}

\newtheorem{assumption}
{Assumption}
\newtheorem{definition}{Definition}[section]
\newtheorem{remark}{Remark}

\newcommand{\expec}{\mathbb{E}}

\newcommand{\beq}{\begin{eqnarray*}}
\newcommand{\eeq}{\end{eqnarray*}}
\newcommand{\beqn}{\begin{eqnarray}}
\newcommand{\eeqn}{\end{eqnarray}}
\newcommand{\ben}{\begin{enumerate}}
\newcommand{\een}{\end{enumerate}}
\newcommand{\bit}{\begin{itemize}}
\newcommand{\eit}{\end{itemize}}
\newcommand{\hide}[1]{}

\newcommand{\Y}{\boldsymbol{Y}}

\newcommand{\eps}{\varepsilon}
\newcommand{\vertiii}[1]{{\left\vert\kern-0.25ex\left\vert\kern-0.25ex\left\vert #1 
    \right\vert\kern-0.25ex\right\vert\kern-0.25ex\right\vert}}

\renewcommand{\epsilon}{\eps}

\newtheorem{lemma}{Lemma}

\newtheorem{theorem}{Theorem}
\newtheorem{corollary}{Corollary}

\begin{document}
%
\title{Goggin's Corrected Kalman Filter: Guarantees and Filtering Regimes}
%
%
%
\author{Imon~Banerjee
        and~Itai~Gurvich%
\thanks{I. Banerjee is with the Department of Industrial Engineering and Management Sciences,
Northwestern University, Evanston, IL 60208 USA (e-mail: imon.banerjee@northwestern.edu).}%
\thanks{I. Gurvich is with the Kellogg School of Management, Northwestern University,
Evanston, IL 60208 USA (e-mail: i-gurvich@kellogg.northwestern.edu).}%
}

%
%

\markboth{Submitted to \textit{IEEE Transactions on Information Theory}}%
{Banerjee and Gurvich: Goggin's Corrected Kalman Filter: Guarantees and Filtering Regimes}

%



\maketitle

\begin{abstract} We revisit a nonlinear filter for linear state-space models with {\em non-Gaussian} signal and observation noise, introduced by Eimear Goggin in \cite{goggin1992convergence}. Goggin showed that applying the Kalman filter (KF) to score-transformed observations is asymptotically optimal in a specific signal-to-noise ratio (SNR) regime. We provide non-asymptotic convergence rates across a broad range of SNR regimes, recovering Goggin's setting as a special case, with bounds explicit in the SNR.

Our analysis combines two ingredients: a posterior Cram\'er--Rao lower bound for filtering and convergence-rate bounds in the Fisher information central limit theorem. We also map the parameters space into filtering regimes, identify degenerate regimes---where simple filters are nearly optimal---and isolate a {\em balanced} regime, where Goggin's filter has the most value.
\end{abstract}


%
\IEEEpeerreviewmaketitle

\section{Introduction} 
We study the basic filtering problem in a linear state-space model: estimate the unobserved state process $(X_t, t \in \mathbb{N})$ from noisy observations $(Y_t, t \in \mathbb{N})$:
\begin{equation}\begin{split}  
X_{t+1} & = \gamma X_t + w_t, ~t\in \{0\}\cup \mathbb{N},\\
Y_t& = X_t+ v_t,~t\in\mathbb{N},\end{split} \tag{Linear System}\label{eq:linearsystemintro}
\end{equation}
where $(w_t)$ and $(v_t)$ are i.i.d. zero-mean random variables. We normalize the units so that $\Var(w_t)=\Var(v_t)=1$. 

At each time $t$, an estimator $\hat{x}_t$ is a mapping from $(Y_1,\ldots,Y_t)$ to $\mathbb{R}$. When $w$ and $v$ are Gaussian, the sequential estimator produced by the Kalman filter (KF) is unbiased (that is, $\Ex[\hat{x}_t]=\Ex[X_t]$) and its mean squared error (MSE), $\mathbb{E}[(\hat{x}_t - X_t)^2]$, achieves the Cram\'er--Rao lower bound (CRLB), making it the minimum variance unbiased estimator (MVUE). If either $w$ or $v$ are non-Gaussian, the KF is no longer MVUE, although it remains the best possible unbiased linear filter (BLUE). 

A natural question is whether the appealing simplicity of the KF can be retained while achieving near minimality of variance in the non-Gaussian case. Goggin \cite{goggin1992convergence} answered this question in the following scaling  \footnote{\cite{goggin1992convergence} uses the observation process $\bar{Y}_t: = (1/N) Y_t= \frac{1}{N}X_t+  v_t/\sqrt{N}$. We equivalently multiply both sides by $N$.}  of (\ref{eq:linearsystemintro}) 
\begin{equation}\begin{split} 
    X_{t+1} & = \lp 1-\frac{1}{N} \rp X_t + \frac{1}{\sqrt{N}} w_t,~t\in \{0\}\cup \mathbb{N},\\
    Y_t& = X_t + \sqrt{N} v_t,~t\in\mathbb{N}.\end{split} \label{eq:goggin_regime}
\end{equation}

Let $h$ denote the density of $v$ and $\phi=-h'/h$ be the negative of the score function. \cite[Section IV]{goggin1992convergence} demonstrates that running the Kalman filter on the modified system
\begin{equation}
\begin{split}
    X_{t+1} & = \lp 1-\frac{1}{N} \rp X_t + \frac{1}{\sqrt{N}} w_t, ~t\in \{0\}\cup \mathbb{N},\\
    Z_t & = \sqrt{N}\phi\lp \frac{Y_t}{\sqrt{N}} \rp, ~t\in \mathbb{N}.\label{eq:goggin_original}
\end{split}
\end{equation}
yields an estimate of the state that is asymptotically optimal as $N\rightarrow \infty$. Since non-Gaussian systems are often estimated via heuristic methods like moment-matching \citep{clark_new_2006,shimizu_moment-based_2023}, the Goggin filter in \cref{eq:goggin_original} provides a simple, provably efficient alternative in this specific signal-to-noise ratio (SNR) regime; see discussion in Section \ref{sec:model}.

In this paper we
\begin{enumerate}
    \item 
     provide non-asymptotic (in $N$) performance guarantees for Goggin's filter. This is our \emph{pre-limit} counterpart to Goggin's asymptotic optimality results, and
\item expand the scaling in \cref{eq:goggin_regime} to a wide range of SNR regimes, and characterize precisely the settings where the Goggin filter \cref{eq:goggin_original} outperforms the Kalman filter.
    
    \end{enumerate}
To establish these results, we first expand Goggin's setting in \cref{eq:goggin_regime}. Specifically, we consider the optimal state estimation for the following generalization of \cref{eq:goggin_regime}:
\begin{equation}\begin{split} 
    X_{t+1} & = \lp 1-\frac{1}{N} \rp X_t + \frac{1}{\sqrt{N}} w_t,~t\in \{0\}\cup \mathbb{N},\\
    Y_t& = X_t + s_N v_t,~t\in\mathbb{N};\end{split} \label{eq:our_regime}
\end{equation} $s_N=\sqrt{N}$ is the special case in \cref{eq:goggin_regime}. Depending on the specific value of $s_N$, we run the Kalman filter on two versions of the the modified (possibly nonlinear) system (see Definitions \ref{def:goggin-filter} and \ref{def:centered-goggin-filter} for full details) \begin{equation}
\begin{split}
    X_{t+1} & = \lp 1-\frac{1}{N} \rp X_t + \frac{1}{\sqrt{N}} w_t, ~t\in \{0\}\cup \mathbb{N},\\
    Z_t & = s_N\phi\lp \frac{Y_t}{s_N} \rp, ~ t\in \mathbb{N}.
\end{split}\tag{Scaled System} \label{eq:goggin_filter}
\end{equation} The dynamics can be further generalized to have $ X_{t+1} = \lp 1-N^{-1} \rp X_t + g_N w_t, ~t\in \{0\}\cup \mathbb{N}$ for $g_N>0$ that is not necessarily $1/\sqrt{N}$. We will show in Section \ref{sec:model} that outside of the regime in which $(1-(1-N^{-1})^2)/g_N^2=\Theta(1)$---which automatically implies $g_N=\frac{1}{\sqrt{N}}(1+\ocal(1))$---simpler filters are nearly optimal. In other words, such generalizations do not yield more insights.

\subsection{Technical Contributions}\label{sec:technical-contrib}
Our main contribution is a theoretical demonstration that the MSE of Goggin's filter as applied to (\ref{eq:goggin_filter}) (see  Definition \ref{def:goggin-filter}) or its centered counterpart, (see Definition \ref{def:centered-goggin-filter}) approximates well the Cram\'er--Rao lower bound (Theorem \ref{thm:main}) in the pre-limit regime of $N<\infty$. Our analysis uses recently-developed tools such as convergence rates for Fisher information (Lemma \ref{lem:CLT_gn}), allowing us to derive the sub-optimality \textit{rate} of the proposed filter. We also characterize three filtering regimes based on the system's signal-to-noise ratio (SNR): high SNR, low SNR, and a balanced regime where the Goggin filter (GF) provably outperforms the Kalman filter (KF). 

\textbf{Filtering Regimes:} At the most basic level, our paper strengthens Goggin's \cite{goggin1992convergence} limit-theorem result by proving that the filter's steady-state convergence {\em rate} is $1/\sqrt{N}$ based on first principles; see Corollary \ref{cor:rateGoggin}. Yet, we go beyond Goggin's limit theorem to cover a wider range of parameter settings, or {\em regimes}. Goggin's analysis---based on a limit theorem for conditional expectations---focuses on the case in \cref{eq:goggin_regime}, where the best possible mean squared error, $MSE^{\star}$, is of the same order of magnitude of the variance of the (unobserved) process: $MSE^{\star}\approx \Ex[X_t^2]$. 

Our Lemmas \ref{lem:toonoisy}, \ref{lem:nonoise} show the sufficiency of trivial filters in the high and low SNR regimes, whereas Lemma \ref{lem:balanced} derives how much one can improve over the Kalman filter in the balanced SNR region. The intuitions for the former lemmas are clear: if the observation noise is small (large) enough, one should use the observation(expected observation) as the estimate of the state. But a direct proof via CRLB is insufficient, necessitating a simple Le Cam-type (minimaxity) analysis that, through a study of the posterior probabilities, shows that one cannot do (much) better than using the ``current'' observation alone. 

Having identified these extreme regimes, we move to the derivation of  pre-limit guarantees for the Goggin filter in the intermediate {\em balanced} SNR regime.

\textbf{Goggin's filter:} 
We study two versions of the filter that differ in the way the score $\phi$ is applied to the observation.  We refer to these as the {\em centered} and {\em non-centered} GF. For both, we prove a lower bound and a  matching upper bound.

The pre-limit view allows us to study a wide range of SNR values for these filters. Together with our mapping of regimes, our results show that GF, together with two trivial filters---taking the observation as the true state and taking the expected state as the true state---suffice for near-complete coverage of the spectrum of regimes.

A simple example conveys the intuition behind Goggin's filter. Fix $\expec[\cdot]$, $\Var(\cdot)$, and $I(\cdot)$ to be the canonical expectation, variance, and Fisher information, respectively, and consider the case where $X_0=x_0$ is deterministic and no process noise is present. This is the problem of estimating a (location) parameter $x_0$ from i.i.d. noisy observations $Y_t = x_0+s_N v_t, t\in\mathbb{N}$. It is a fact that the average 
$$\hat{x}^A_t= \frac{1}{t}\sum_{s=1}^{t}Y_s,$$
is unbiased but not MVUE unless $v_t$ is Gaussian (\cite{bickel2015mathematical}). Recall that $\Var(v_1)=1\geq I(v_1)^{-1}$ with equality holding if and only if $v_1$ is Gaussian.
In this case, $\hat{x}^A_t$ achieves the Cram\'er--Rao lower bound 
$$\Var(\hat{x}^A_t) = \frac{s_N^2}{t} = \frac{s_N^2}{t I(v_1)} = CRLB.$$ 
Now, let $v_1$ be \textit{non-Gaussian} and consider the estimator
$$\hat{x}_t = \frac{s_N}{tI(v_1)}\sum_{s=1}^t \phi(Y_s/s_N).$$
Since $Y_t$'s are i.i.d. $\hat x_t$ has expectation and variance given by 
$$\Ex[\hat{x}_t] = \frac{s_N}{I(v_1)}\Ex[\phi(Y_1/s_N)], ~~~\Var(\hat{x}_t) = \frac{s_N^2}{tI^2(v_1)}\Var(\phi(Y_1/s_N)).$$ 
Informally applying a Taylor expansion to $\phi(Y_1/s_N)$ we have 
$$\phi(Y_1/s_N) = \phi(  v_1+x_0/s_N) \approx \phi(  v_1) +\phi'(v_1)x_0/s_N.$$ 
Because $(-\phi)$ is the score function of $v$, $\Ex[\phi(v_1)]=0, \Ex[\phi'(v_1)] = \Var(\phi(v_1))= I(v_1)$ so that (ignoring higher order $1/s_N^2$ terms)
$$ \Ex[\hat{x}_t]\approx x_0, \ \Var(\hat{x}_t) \approx \frac{s_N^2}{tI(v_1)}=CRLB.$$ 
Hence, this estimator {\em approximately} achieves the CRLB with a bias and variance suboptimality  that depend on the remainder in the first-order Taylor expansion.

Goggin's analysis (\cite[Theorem 2]{goggin1992convergence}) is grounded in a change of measure (likelihood) view followed by convergence of conditional expectations. To go beyond limit theorems our proofs build a bridge between two literature streams---CRLB for sequential estimation and Fisher Information CLT---to produce a useful lower bound on the best possible stationary MSE (Theorem \ref{thm:CRLB}). The Fisher information CLT justifies  aggregating the signal noise across periods. Plugged into the CRLB for sequential estimation (see \cref{eq:tichavsky}), the aggregation produces what we show to be a suitably tight lower bound. 

The observation noise cannot be aggregated without significant compromise to this tightness. This is because in the ``balanced'' regime, where the Goggin filter outperforms KF (Lemma \ref{lem:KF}), the CLT error is too big for the observation noise but ``just right'' for the signal noise.



\section{Literature review\label{sec:lit}} 

We consider partially observed linear systems with Gaussian or non-Gaussian signal and observation noise. When both noises are Gaussian, the Kalman filter \cite{kalman1960new} is optimal, minimizing the mean squared error at each time $t$. 
With non-Gaussian signal or observation noise, the KF is suboptimal and the MSE-minimizing filter can be substantially more complex. A vast literature explores strategies to handle non-linearity or non-Gaussianity, offering a wide range of heuristic approaches. For background and general treatments, see, e.g., \cite{sarkka2023bayesian,krishnamurthy2016partially}. 
Our focus is on the non-Gaussianity of the signal and observation noises; see \cite{Kitagawa1987} for early work that emphasizes non-Gaussianity noise. Closest to our work is the literature on asymptotically optimal filters.

In \cite{liptser1996robust} the authors consider the filtering of a continuous space and time diffusion process observed at discrete times. They construct a filter by introducing a transformation to the observations \`a la Goggin, and proceed to establish the asymptotic optimality of the filter as the time between observations and the observation variance shrink; \cite{Bobrovsky2001} follows up on \cite{goggin1992convergence} and considers filtering of a continuous diffusion process. Using a suitable CRLB, the authors show that the Goggin filter is asymptotically optimal. Their transformation of the observation, in contrast with Goggin's, is a centered transformation; the paper \cite{Liptser1998}  compares numerically the centered and non-centered transformations. They refer to the transformation as an application of a ``limiter''; see also \cite[Chapter 20]{Lipster1997}. 

Importantly, because the signal process
in \cite{Bobrovsky2001,Liptser1998}
is a diffusion process, the corresponding discrete-time model is one with $w_t$ (in our notation) being Gaussian. We work directly with the pre-limit discrete-time model and prove that applying a ``fake'' KF to \`a-la-Goggin transformed data produces a nearly optimal estimator for a wide spectrum of observation-noise levels (as captured in the model by $s_N$), where the optimality gap is explicitly quantified. 
 
  Instead of being based on functional central limit theorem and convergence of conditional expectations, our arguments are based on first principles that are available in the literature: (i) a Cram\'er--Rao lower bound for filtering, and (ii) convergence rate results for the Fisher information as laid developed in \cite{johnson2004information}. 

To apply the Fisher-information CLT, we adopt the assumption from \cite{johnson2004information} that the random variables in the random walk (in our case the signal noises) have a finite restricted Poincar\'e constant (Assumption \ref{asum:primitives-driving}). This condition has been recently relaxed \cite{bobkov_fisher_2014}. 
\vspace*{0.2cm}

\noindent {\bf Notation and organization.} A random variable $X$ is defined with respect to a usual filtered probability space $(\Pcal,\Fcal)$ with $\expec[X]$, $\Var[X]$, and $Cov(\cdot)$ denoting the canonical expectation, variance, and covariance. By $I(X)$, we denote the Fisher information (or  simply information) of the random variable $X$. If $\{X_1,X_2,\dots\}$ is a sequence of identically distributed random variables, $\expec [X], \Var(X)$, and $I(X)$ shall denote its common mean, variance, and information respectively. To simplify notation we sometimes write $\Sigma(X) := Var(X)$ and $\Sigma^{-1}(X) = \frac{1}{Var(X)}.$ In (\ref{eq:linearsystemintro}) the observation $Y_k$ belongs to a location family with the location being the random parameter $X_k$. We repeatedly use the facts (see, for instance, \cite{barron_entropy_1986}) that $I(aX)=I(X)/a^2$ when $X$ belongs to a location family, and $\Var(X)I(X)\geq 1$ with equality holding if and only if $X$ is Gaussian (which forms the basis of our CLT arguments). 

We indicate asymptotic order
following the usual Bachmann--Landau notation for   $\Ocal(\cdot),\Theta(\cdot), \ocal(\cdot)$ \cite{cormen_introduction_2022}. For non-negative functions $f,g$, we sometimes write $f(N)\ll g(N)$ for $f(N)=\ocal(g(N))$ and $f(N)\lesssim g(N)$ for $f(N) = \Ocal(g(N))$. 

For any integer $\tau$, we define $[\tau]:=\{1,\dots,\tau\}$ so that $\sum_{s=0}^{\tau-1}$ can be written $\sum_{s\in [\tau]}$. Finally, we introduce $\gamma:=(1-N^{-1})$ for convenience. 

The remainder of the paper is divided into three main parts: sections \ref{sec:model} and \ref{sec:GF} contain the main results: Section \ref{sec:model} identifies ``degenerate regimes'' and trivial filters, and Section \ref{sec:GF} states the convergence-rate results for the Goggin filter. Sections \ref{sec:lowerbound}-\ref{sec:GFcenteredUB} contain the main proofs, while the appendix collects proofs of auxiliary lemmas.  

\section{Filtering regimes\label{sec:model}} 

Given the distributions of $w_t,v_t$, the family (in $N$) of dynamical models in \eqref{eq:our_regime} is fully specified by the sequence of parameters $\gamma=\gamma_N$ and $s_N$. 
We partition the space of possible parameters into filtering regimes based on the relationship between $N$ and $s_N$. We fix $\gamma_N=(1-1/N)$ to normalize the process, i.e., so that $\Var[X_t]=\Theta(1)$ in stationarity. 

The best possible mean squared error, $MSE^{\star}$ is proportional to $\min\{s_N/\sqrt{N},1\};$ see Theorem \ref{thm:main}. This characterization gives a natural partition into SNR regimes. 

One regime is where the observation noise is (much) bigger than the signal noise
\[ 
s_N\gg \sqrt{N};\tag{Negligible SNR}
\] equivalently $s_N/\sqrt{N}\gg 1$. Intuitively, when the noise is excessively large, the observations are rendered uninformative and the trivial estimator $\hat{x}_t=\Ex[X_t]$ is nearly optimal (Lemma \ref{lem:toonoisy}). The proof require care as the CRLB is loose here; see Section \ref{sec:prf-regimelemmas} for full details.

This looseness of the CRLB extends to the regime where the observation noise is very small, 
\[ 
s_N\ll \frac{1}{\sqrt{N}};\tag{Large SNR}
\] equivalently $MSE^{\star}\ll 1/N$ which is the variance of the driving-noise. Intuitively, with small noise the observation is as good as any other estimator. Formally, (see Lemma \ref{lem:nonoise}) in this regime, taking the actual observation as the estimator---$\hat{x}_t=Y_t$ is nearly optimal. 

Goggin's setting is on the upper boundary of, what we call, the balanced signal-to-noise ratio (SNR) regime: 
\begin{equation}
\frac{1}{\sqrt{N}}\lesssim s_N\lesssim \sqrt{N}. \tag{Balanced SNR}
\end{equation} in \cite{goggin1992convergence} and related work (see Section \ref{sec:lit}) $s_N\approx \sqrt{N}$.  

In the next section (Theorems \ref{thm:main}  and \ref{thm:maincentered}) we show that for most of the balanced regime---not only on its upper boundary---the sub-optimality induced of the GF is order-of-magnitude smaller than the best achievable MSE.

Left out of our analysis is only the subset of the balanced regime where the SNR is small but not negligible, namely  
\begin{equation}
\frac{1}{\sqrt{N}}\lesssim s_N\lesssim 1.\tag{Low SNR} 
\end{equation} 

In this window, our bounds for the GF are no better than those for the KF. Our analysis suggests that this is because, in this regime, the first-order Taylor series approximation---which is central to the GF arguments---is too loose. 

Let $MSE^{\star}$ be the best possible steady-state MSE across all estimators $\hat{x}_t=\pi({\bf Y}_t)$ where ${\bf Y}_t=(Y_1,\ldots,Y_t)$. That is, 
\begin{equation} \label{eq:MSEdefin} MSE^{\star} =  \inf_{\pi}\liminf_{t\uparrow \infty}\Ex[(\pi({\bf Y}_t)-X_t)^2].\end{equation}

\noindent Our first order of business is to characterize the degenerate regimes: Negligible and Large SNR. 

To introduce our key assumptions, we must define the restricted Poincar\'e constant of the random variable $w$: 
\begin{equation} \label{eq:poincaredefin} R^{\star}(w) = \sup_{g\in H_1^*(w)}\frac{\Ex[g^2(w )]}{\Ex[(g'(w))^2]},\end{equation} 
where 
$H_1^*(w)$ is the space of absolutely continuous functions $g$ such that $\Var(g(w))>0$, $\Ex[g(w)]=0$, $\Ex[g^2(w)]<\infty$ and $\Ex[g'(w)]=0$. For Gaussian $w$ this constant is precisely $1/2$. 

\begin{assumption}[signal-noise distribution] \label{asum:primitives-driving} The signal-noise random variable $w$ has a continuously differentiable density on the real line and a finite restricted Poincar\'e constant: $R^{\star}(w)<\infty$. 
\end{assumption} 

\begin{assumption}[observation-noise distribution] The variable $v$ has a finite fourth moment. Its density $h$ is four-times continuously differentiable with $h>0$ on the real line. Also, $\phi=-h'/h$ is integrable, in which case, $\Ex[\phi(v)]=0$ and $\Var(\phi(v))=\Ex[\phi'(v)]=I(v)$. Finally, we assume that  $\|\phi''\|_{\infty}<\infty$. \label{asum:primitives-observation}\end{assumption} 

\noindent {\bf Assumptions \ref{asum:primitives-driving} and \ref{asum:primitives-observation} are imposed throughout.} We refer to these as our {\em basic assumptions}.

\begin{lemma}[Negligible SNR] If $s_N \gg \sqrt{N}$ then $$MSE^{\star} \geq 1/2 -\ocal(1).$$ The lower bound is achieved by the trivial estimator $\hat{x}_t^1 = \Ex[X_t]$.   
\label{lem:toonoisy}
  \end{lemma} 
  
\noindent All lemma proofs appear in the appendix. 
\begin{remark} It is known \cite[Theorem 2]{borovkov1984inequality} that $R^{\star}(w)<\infty$ implies the existence of all moments for $w$. This strong requirement allows us to use the Fisher-information CLT for independent but {\bf not} identically-distributed random variables in   \cite{Johnson2004,johnson2004information}, as quoted in our Lemma \ref{prop:johnson_inforbd}.  
The requirement of a finite Poincar\'e constant is relaxed in \cite{Bobkov2014,Johnson2020}, for the case of i.i.d. random variables, but to the best of our knowledge, no equivalent relaxation exists for the non-i.i.d. case. Any such relaxation would immediately apply to our own results. \hfill \bsq \vspace*{0.2cm}
\end{remark}

Lemma \ref{lem:toonoisy} is conceptually intuitive: what better can one do if the noise is too big than take the expectation as the estimate. But proving this is not as straightforward as showing that the CRLB is asymptotically achieved. This is because the CRLB turns out to be too loose. The proof relies on a CLT for the Fisher information that allows us to replace the information of the non-Gaussian $w$ with the reciprocal of its variance. It is a tool we use again in section \ref{sec:lowerbound}. 

\begin{lemma}[Large SNR] Suppose that $s_N=\ocal(1/\sqrt{N})$ and that, in addition to the basic assumptions, the density of ${w}$ is strictly positive and Lipschitz continuous on the real line. Then, $$MSE^{\star} \geq s_N^2 - \ocal(s_N^2)$$ 
with the lower bound achieved by the trivial estimator $\hat{x}_t = Y_t$.
\label{lem:nonoise}   
\end{lemma}

\begin{remark}
    The requirement in Lemma \ref{lem:nonoise} that $w$ has a Lipschitz continuous density simplifies the analysis. We do not prove that it is necessary. \hfill \bsq \vspace*{0.2cm}  
\end{remark}
The intuition is, again, clear: if the observation noise is small enough, one should use the observation as the estimate of the state. The proof is not direct because the CRLB is too loose in this setting. Instead, the proof relies on a Le Cam--type analysis that, through a study of the posterior probabilities, shows that one cannot do better than using the ``current'' observation alone. 

In the intermediate (balanced) regime, though strictly sub-optimal (see Lemma \ref{lem:KF}), the Kalman filter improves on the trivial filters in Lemmas \ref{lem:toonoisy} and \ref{lem:nonoise}.

\begin{lemma}[Balanced Regime] Suppose that $\frac{1}{\sqrt{N}}\lesssim s_N\lesssim \sqrt{N}$. Then both estimators in Lemmas \ref{lem:toonoisy} and \ref{lem:nonoise} are strictly sub-optimal: 
$$MSE\pow 1 - MSE^{KF}  = \Omega  (MSE\pow 1), ~~ MSE\pow 2- MSE^{KF}=\Omega(MSE\pow 2),$$
where $MSE^{KF}$ is the steady-state MSE of the Kalman filter, and $MSE\pow 1$ and $MSE\pow 2$ are the steady-state MSE of the trivial filters in Lemmas \ref{lem:toonoisy} and \ref{lem:nonoise} respectively. 

\label{lem:balanced} 
\end{lemma}

\section{The GF convergence rate in the balanced regime\label{sec:GF}} 

Recall (Assumption \ref{asum:primitives-observation}) the negative of the score function $\phi = -h'/h$ which satisfies 
$$\Ex[\phi(v)]=0, \mbox{ and } \Var(\phi(v))=I(v),$$ where $I(v)$ is the Fisher information of the distribution $h$ (sometimes referred to as the location information). We similarly $I(w)$ for the information of $w$. 

\textbf{Goggin filter:} The well-known Kalman filter produces its state estimator, $\hat{x}_t$, via the recursion \begin{equation}\tag{Kalman Filter} \hat{x}_t = \gamma\hat{x}_{t-1} + K_t(Y_t-\gamma\hat{x}_{t-1}),\end{equation} where $K_t$ is the so-called Kalman gain. 

To take care of non-Gaussianity, the setup for Goggin filter starts with the nonlinear transformation of the observations $Y_t$ and we consider two such transformations. 

The first, consistent with the original filter of Goggin, has the update 
\begin{equation}\tag{GF} \hat{x}_t = \gamma\hat{x}_{t-1} + K_t\left(s_N\phi\left(\frac{Y_t}{s_N}\right)-I(v)\gamma \hat{x}_{t-1}\right); \end{equation} 
see Definition \ref{def:goggin-filter}. This transformation is applied to the observation {\em pre-centering}. In the next variant, the transformation is applied to the centered observation.  
\begin{equation}\tag{Centered GF} \hat{x}_t = \gamma\hat{x}_{t-1} + K_t s_N\phi\left(\frac{Y_t-\gamma\hat{x}_{t-1}}{s_N}\right);\end{equation} 
see Definition \ref{def:centered-goggin-filter}. If $v$ is Gaussian, both version of the GF reduce to the Kalman filter. As in Goggin's \cite{goggin1992convergence} analysis, we transform the {\em observation} noise but do not modify the filter in any way that accounts for the non-Gaussianity of the signal-noise sequence $(w_t)$. For that noise, an aggregation effect, we will prove, ensures that the approximation gap introduced by ignoring this non-Gaussianity remains suitably small.

\begin{definition}[Goggin Filter ({\bf GF})]\label{def:goggin-filter} The estimator is given by the recursion:
$$\hat{x}_t = \gamma\hat{x}_{t-1} + K_t\left(s_N\phi\left(\frac{Y_t}{s_N}\right)-I(v)\gamma \hat{x}_{t-1}\right),\footnote{The distribution of $\hat{x}_0$ is immaterial as we consider stationary performance.} $$
where, denoting $R= s_N^2 I(v),Q=Var(w_t/\sqrt{N}) =1/N$, \begin{equation} \label{eq:KPdefin} K_t=\frac{P_tI(v)}{R},~\mbox{where $P_t$ follows the recursion } P_t = \frac{R(\gamma^2 P_{t-1} + Q)}{I^2(v)(\gamma ^2 P_{t-1}+Q)+R}.\end{equation}   
\end{definition}

The unobserved process $X_t$ has a steady-state distribution because $\gamma \in (0,1)$. When we say below that the MSE of GF, $\Ex[(\hat{x}_t-X_t)^2]$, achieves an upper bound {\em in stationarity} we mean that $\limsup_{t\uparrow} \Ex[(\hat{x}_t-X_t)^2]$ satisfies the stated upper bound; the same applies to the bias bounds. 

\begin{theorem} Suppose $1<< s_N \lesssim \sqrt{N}$ and let $\hat{x}_t$ be the (random) estimator produced by Goggin filter. Then, under the basic assumptions in stationarity, 
\begin{align} 
|\Ex[\hat{x}_t-X_t]| &= \Ocal\left(\frac{1}{s_N}\right) \tag{Bias} \\
\Ex[(\hat{x}_t-X_t)^2]& = MSE^{\star}\pm \Ocal\left(\frac{1}{s_N}\right),\tag{Variance} \end{align} where $MSE^{\star}$ is the best possible $MSE$ in \cref{eq:MSEdefin}. Moreover, $MSE^{\star} = \Theta(s_N/\sqrt{N})$ so that 
$$\frac{\Ex[(\hat{x}_t-X_t)^2]}{MSE^{\star}}= 1 + \Ocal\left(\frac{\sqrt{N}}{s_N^2}\right).$$\label{thm:main}
\end{theorem}

\begin{corollary}[Goggin's setting $s_N=\sqrt{N}$] Suppose $s_N=\sqrt{N}$ and let $\hat{x}_t$ be the (random) estimator produced by the Goggin filter. Then, in stationarity, both
\begin{align*} 
|\Ex[\hat{x}_t-X_t]| = \Ocal\left(\frac{1}{\sqrt{N}}\right),~  \mbox{and } 
\Ex[(\hat{x}_t-X_t)^2] = MSE^{\star}\pm \Ocal\left(\frac{1}{\sqrt{N}}\right), \end{align*} so that 
$$\frac{\Ex[(\hat{x}_t-X_t)^2]}{MSE^{\star}}= 1 + \Ocal\left(\frac{1}{\sqrt{N}}\right).$$
\label{cor:rateGoggin}
\end{corollary} 

While Theorem \ref{thm:main} covers $1<< s_N\lesssim \sqrt{N}$, the result is less meaningful where  $ s_N\lesssim N^{1/4}$. In this range, the MSE percent gap is a constant and hence not negligible. \cite{Liptser1998} show, numerically, that---in the case $s_N=\sqrt{N}$---
a centered ``limiter'' has lower MVUE. Theorem \ref{thm:maincentered} below provides a pre-limit justification for this fact. It establishes that a centered version of GF has negligible percent gap over the full range $1<< s_N\lesssim \sqrt{N}$. For this, however, we do require, stronger assumptions on the observation-noise distribution. 

\begin{definition}[Centered Goggin Filter ({\bf GF})]\label{def:centered-goggin-filter} The estimator is given by the recursion:
$$\hat{x}^{c}_t = \gamma \hat{x}^{c}_{t-1} + K_{t}s_N\phi\lp \frac{Y_t- \gamma \hat{x}^{c}_{t-1}}{s_N} \rp$$
where $K_t$ is as in \eqref{eq:KPdefin}. 
\end{definition}

\begin{assumption}\label{asum:dissipativity}
Let $\phi = - f'/f$ denote the score of the measurement noise. 
We assume:

\begin{enumerate}
\item \textbf{Strong dissipativity.}  
There exists $\zeta > 0$ such that the function 
$h(y) := \expec[\phi(v_t + y)]$ satisfies \begin{align*}
h(y)\,y \;\ge\; \zeta y^2, ~\forall y \in \mathbb{R}.\numberthis\label{eq:dissipative}
\end{align*}

\item \textbf{Linear growth.}  There exist constants $A,B<\infty$ such that 
\begin{align*}
|\phi(y)| \;\le\; A + B|y|, ~\forall y \in \mathbb{R}.  \numberthis\label{eq:linear_growth}
\end{align*}
\end{enumerate}\label{asum:centered} 
\end{assumption}

The term strong-dissipativity relates to the notion of dissipativity in dynamical systems that guarantees stability. As the proof of Lemma \ref{lem:preliminary} shows, this condition guarantees that the variance recursion for the filter is stable. The linear growth condition gives control on the one-step growth in that same variance recursion. The following lemma characterises one family of distributions that satisfy Assumption \ref{asum:centered}. 

\begin{lemma}\label{lemma:log-concave} Any density of the form $f(x) =\frac{1}{Z}e^{-V(x)}$ for $V$ that is three time differentiable on the real line and has: (i) $\inf_{x} V''(x) >0$, (ii) $|V'(x)| \leq A+ B|x|$, and (iii) $|V'''(x)|\leq C$ satisfies Assumptions \ref{asum:primitives-observation} and \ref{asum:centered}. 
\end{lemma}

\begin{theorem} \label{thm:maincentered}
Suppose $1<< s_N\lesssim \sqrt{N}$ and add Assumption \ref{asum:centered} to our basic assumptions. Let  $\hat{x}^{c}_t$ be the (random) estimator produced by the centered Goggin filter. Then, in stationarity, 
\begin{align} 
\Ex[\hat{x}^{c}_t-X_t] &= \Ocal\left(\frac{1}{\sqrt N}\right) \tag{Bias} \\
\Ex[(\hat{x}^{c}_t-X_t)^2]& = MSE^{\star}\pm \Ocal\lp\frac{\sqrt{s_N}}{N^{3/4}}\rp
,\tag{Variance} 
\end{align} 
where $MSE^{\star}$ is the best possible $MSE$ in \cref{eq:MSEdefin}. Moreover, $MSE^{\star} = \Theta(s_N/\sqrt{N})$ so that 
$$\frac{\Ex[(\hat{x}^{c}_t-X_t)^2]}{MSE^{\star}}= 1 + \Ocal\left(\frac{1}{\sqrt{s_N}N^{1/4}}\right).$$
\end{theorem}

\vspace*{.5cm} 
\noindent{\em Asymptotic optimality:} Observe that a straightforward implication of Theorems \ref{thm:main} and \ref{thm:maincentered} is that, under the suitable assumptions, steady-state asymptotic optimality for:
$$\mbox{ Bias } \rightarrow 0,\mbox{ and } \frac{MSE}{MSE^{\star}}\rightarrow 1, \mbox{ as }N\uparrow \infty.$$

\begin{remark}[Filter comparison for $s_N=\sqrt{N}$]
\noindent On the boundary of the balanced regime, where $s_N=\Theta(\sqrt{N})$, both variants---centered and non-centered---produce the same MSE sub-optimality rate of $1/\sqrt{N}$. In other words, to the extent that the centered variant improves on the non-centered one, this improvement must be a second-order one. \hfill  \bsq \end{remark} 

\begin{remark}[On the scaling of $MSE^{\star}$]\label{remark:KFvsGF}
Achieving the optimal scaling of $MSE^{\star}$ which is $\Theta(s_N/\sqrt{N})$ for $1\lesssim s_N\lesssim \sqrt{N}$ is not where adaptive filters---in this case KF vs. GF---differ. Instead, they differ only in how close they get to the correct constant multiple of $s_N/\sqrt{N}$. Indeed, a naive averaging  filter already achieves the optimal $MSE^{\star}$ scaling. The {\textbf{naive}} filter is defined as follows: for all $t \in [(k-1)\tau+1,k\tau]$ we set
\begin{align*}
\hat{x}_t = \frac{1}{\tau}\sum_{s=(k-1)\tau+1}^{k\tau}Y_s.    
\end{align*}
The MSE of this estimator is given by $$\Ex[(\hat{x}_t-X_t)^2] = \frac{1}{\tau^2}\Ex\left[\left(\sum_{s=(k-1)\tau+1}^{k\tau}(X_s-X_t)\right)^2 \right] +\frac{1}{\tau}\Var(s_Nv_1)= \Theta\left(\frac{\tau}{N}\right) +  \frac{1}{\tau}s_N^2$$ 
To achieve the $MSE^{\star}$ scaling of $s_N/\sqrt{N}$ with this naive filter, a batch size $\tau=\Theta(s_N\sqrt{N})$ suffices. \hfill \bsq \vspace*{0.2cm}
\end{remark} 
 


\noindent {\em Improvement over the Regular Kalman filter:} While implicitly clear from earlier literature, the next lemma formally establishes that the Goggin filter offers a substantial improvement over KF. Specifically, for $s_N\gg N^{1/4}$ GF has an error that is $\ocal(MSE^\star)$ but the KF has an error that is $\Omega(MSE^\star)$.

\begin{lemma} Suppose $1\ll s_N\lesssim \sqrt{N}$. Let $MSE^{KF}$ be the steady-state MSE of the Kalman filter. Then, 
 $$MSE^{KF} \geq MSE^\star + \Omega\left(MSE^\star\right)$$
\label{lem:KF}
 \end{lemma} 

\noindent {\em Organization of the proof of Theorem \ref{thm:main} and \ref{thm:maincentered}.} We divide the proofs into two parts. In Section \ref{sec:lowerbound}, specifically Theorem \ref{thm:CRLB}, we establish a lower bound on $MSE^{\star}$. To do that, we fix a batch period $\tau$ and prove the lower bound for the estimates $\hat x_{\tau},\hat x_{2\tau},\dots,$ at batch periods. We then extrapolate from batches to any time $t$ in stationarity. Our strategy for proving the lower bound in the batch periods is to apply Johnson's \cite{johnson2004information} inequality for the information of averages (Proposition \ref{prop:johnson_inforbd}) to replace the information of the dynamics noise with the reciprocal of the variance, {\em as if} this noise is Gaussian. We plug this replacement  into the CRLB for filtering as developed by \cite{tichavsky1998posterior}. This produces a ``near'' CRLB for the batched linear dynamical model. Proposition \ref{prop:johnson_inforbd} is where boundedness assumptions on the restricted Poincar\'e constants (Assumption \ref{asum:primitives-observation}) are used. 

The upper bound on the $MSE$ of GF is developed from first principles in Sections \ref{sec:upperbound} and \ref{sec:GFcenteredUB}. We start with the Goggin filter and---after a Taylor series expansion of the transformed observations $\phi(v_t)$ $s_N\phi(Y_t/s_N)$ (or $s_N\phi((Y_t-\gamma\hat{x}_{t-1}^c)/s_N)$ in the centered version)--- derive upper bounds for the remainder terms. Our ingredients for this analysis are Lemma \ref{lem:KFPXg} which provides an upper bound on $P_t$ and the Kalman gain $K_t$, and Lemmas \ref{lem:remainderterms1} and \ref{lem:remainderterms2}, which provide upper bounds on additional cross terms arising from the Taylor expansion. Lemma \ref{lem:preliminary} produces preliminary crude bounds on the MSE of the centered GF.

\section{The Cram\'er--Rao Lower bound at batch periods\label{sec:lowerbound}}

The CRLB is not tight in Goggin's regime when applied directly. This is because the CRLB retains the information for the observation noise---see \cite{goggin_convergence_1992}--- whereas the Goggin's filter, which we know to be asymptotically optimal,``treats'' the signal noise as Gaussian. We improve the lower bound through batching. 

We study the dynamics of $X_k^B:=X_{\tau k}$, which is the linear system at times $k\tau,k\in\mathbb{N}\cup \{0\}$ satisfying for $k\geq 1$ the linear dynamics
\begin{align*} 
X_{k}^B & = \gamma^{\tau} X_{k-1}^B+\frac{1}{\sqrt{N}}\sum_{s\in [\tau]}\gamma^{s-1} w_{\tau k-s}\\
& = \gamma^{\tau} X_{k-1}^B + \Wcal_k,
\end{align*} 
where $\gamma=(1-1/N)$ and
\begin{align}
    \Wcal_k :=\frac{1}{\sqrt{N}} \sum_{s\in [\tau]}\gamma^{s-1} w_{\tau k-s}, 
    \quad \mbox{so that}\quad  \Sigma(\Wcal_k):=\Var(\Wcal_k) = \frac{1}{N} \frac{1-\gamma^{2\tau}}{1-\gamma^2}.
    \label{eq:def-Wcal}
\end{align}
We initialize $X_0^B := X_0$ and rearrange \cref{eq:our_regime} to write the observations in terms of $X_k^B$. 

Let  $Y_k^B = \lp \gamma^{\tau -1} Y_{\tau (k-1)+1},\dots,\gamma^{\tau k -t}Y_t ,\dots, Y_{\tau k}\rp$ denote the vector of observations between time points $\tau (k-1)+1$ to $\tau k$. Then
\begin{align*}
    Y_k^B = {\bf e} X_k^B + \Vcal_k\in\mathbb{R}^{\tau},
\end{align*}
where ${\bf e}$ is the column vector of $1$'s and 
$\Vcal_k$ has entries 
$$\Vcal_{k,t}=\gamma^{\tau k-t}s_N v_t-\frac{1}{\sqrt{N}} \sum_{s=1}^{\tau k-t}\gamma^{s-1} w_{\tau k-s}, 
~~~ t\in \{\tau(k-1)+1,\ldots,\tau k\},$$ 
and we define $\sum_{s=1}^0=0$. In the following table we summarize the notation that are used \textit{exclusively} in this section for easy reference. 
\begin{table*}[htpb]
    \centering
\begin{tabular}{|c|l|c|l|}
\hline
\textbf{Symbol} & \textbf{Definition} & \textbf{Symbol} & \textbf{Definition} \\
\hline
$\mathbf{w}_k$ & $(w_{\tau(k-1)+1}, \dots, w_{\tau k-1})$ 
&  $\mathbb{V}_k^{-}$ & $s_N \times (\gamma v_{\tau(k-1)+1}, \gamma^2 v_{\tau(k-1)+2}, \dots, \gamma^{\tau-1} v_{\tau k-1})$ \\

$\mathcal{W}_k$ & $\dfrac{1}{\sqrt{N}} \sum_{s \in [\tau]} \gamma^{s-1} w_{\tau k - s}$ 
& 
$\mathbf{v}_k$ & $s_N \mathbb{V}_k^{-}$\\

$\mathbb{W}_k$ & $\left(-\dfrac{1}{\sqrt{N}} \sum_{s=1}^{\tau k - t} \gamma^{s-1} w_{\tau k - s}, ~ t \in [\tau(k-1)+1, \ldots, \tau k-1]\right)$ 
& $\mathbb{V}_k$ & $(\mathbb{V}_k^{-},\, s_N v_{\tau k})$ \\

$\mathbf{e}_k$ & $\tau$-length vector with $1$ in the $k$-th coordinate and $0$ elsewhere 
& $\mathbf{e}$ & $\tau$-length vector of $1$'s \\
\hline
\end{tabular}\label{tab:placeholder}
\vspace{1pt  }\caption{Notation used in the lower bound proof}
\end{table*}
We end up with the system $(X_k^B, Y_k^B)$ which is an equivalent reformulation of  \cref{eq:our_regime} and follows the dynamics 
\begin{align} \label{eq:batch1} 
X_k^B & = \gamma^{\tau} X_{k-1}^B + \Wcal_k
\\Y_k^B &=  {\bf e} X_k^B + \Vcal_k.\label{eq:batch2} 
\end{align}
\cite{tichavsky1998posterior} develop a recursion to compute an information lower bound on the MSE for general non-linear filtering problems; see equation (21) there and additional references in \cite{Bergman2001}. Specialized to our setting, \cite[Proposition 1]{tichavsky1998posterior} provides a recursion for a lower bound for estimating $X_k^B$ from the observations $Y_k^B$. From there it follows that $MSE_k^{*,B}\geq J_k^{-1}$ where $J_k$ which  satisfies the recursion 
\begin{equation}
    \begin{split}
        J_{k+1} & = I(\calW)+{\bf e}' I(\calV){\bf e}-\gamma^{2\tau} I^2(\calW)(J_k+\gamma^{2\tau}I(\calW))^{-1}\\
& = {\bf e}'I(\calV){\bf e}+\frac{J_kI(\calW)}{J_k+\gamma^{2\tau}I(\calW)}, \label{eq:tichavsky}
    \end{split} 
\end{equation} 
where recall from the notations section that $I(\calW)=I(\calW_1)$, $I(\calV)=I(\calV_1)$ and we have set $J_0=I(X_0)$. 
Since for all $k$, 
$$MSE^{\star}_k \geq \frac{1}{J_k},$$
we now only need to approximate $J_k$. This is achieved by the following theorem, which proves that the Fisher information can be replaced by the reciprocal of the variance---as if these noises are Gaussian. 

Below and onward 
$$\Sigma(\calW):= Var(\calW), ~~\Sigma^{-1}(\calW) = \frac{1}{Var(\calW)}.$$

Also 
$MSE^{\star,B}$ is the best possible steady-state MSE for the dynamics in \eqref{eq:batch1}-\eqref{eq:batch2}. That is, 
$$ MSE^{\star,B} = \inf_{\pi}\liminf_{k\uparrow\infty}\Ex[(\pi({\bf Y}_k^B)-X_k^B)^2].$$ 
Then by bounding  $MSE_k^{\star}:=\inf_{\pi}\Ex[(\pi({\bf Y}_k^
B)-X_k^B)^2]$ from below and using $MSE^{\star,B}\geq \liminf_{k}MSE_k^{\star}$, we establish the lower bound in Theorem \ref{thm:CRLB}.

\begin{theorem} Suppose $s_N=\Omega(1)$. For $\tau=\Ocal(N)$, 
\[ 
MSE^{\star,B} \geq \frac{1}{\bar{J}_{\infty}} - \Ocal\left(\frac{1}{\tau \bar{J}_{\infty}}\right)\numberthis\label{eq:MSELB}
\]
where $\bar{J}_\infty$ is the stationary point of the recursion 
\begin{equation}  
\bar{J}_{k+1} = {\bf e}'I(\mathbb{V}){\bf e} + \frac{\bar{J}_k\Sigma^{-1}(\calW)}{\bar{J_k}+\gamma^{2\tau}\Sigma^{-1}(\calW) } \quad \text{and has}\quad \bar{J}_{\infty} = \Theta\left(\max\left\{\frac{\tau}{s_N^2}+\frac{\sqrt{N}}{s_N},1\right\}\right).\label{eq:barJrecursion}\end{equation} 
Finally, if $w$ is normally distributed there is no approximation in \eqref{eq:MSELB}: $MSE^{\star} \geq \bar{J}_{\infty}^{-1}$. 

\label{thm:CRLB}
\end{theorem} 
\vspace*{0.4cm} 

The corollary below states that the lower bound holds not only for the ``batched'' times but holds in stationarity for the period-by-period MSE. 

\begin{corollary}\label{cor:CRLB}
Let $MSE_t^{\star}$ be the best possible MSE for estimating $X_t$. Then, for $s_N=\Omega(1)$ and for $1\ll \tau\lesssim N$,
$$ MSE_{\infty}^{\star} \geq \frac{1}{\bar{J}_{\infty}} - \Ocal\left(\frac{1}{\tau \bar{J}_{\infty}}\right),$$ 
where $\bar{J}_{\infty}$ is as in Theorem \ref{thm:CRLB}.
\end{corollary}

\begin{proof}
In our construction we can take the first batch to be of a different size. Namely, fix $m\in [\tau+1]$ and let $X_1^B = X_m$, For $l\geq 1$, let $X_l^B= X_{m+(l-1)\tau}$. For each fixed $m$, it holds 
$J_{\infty,m}$ is the same as $J_{\infty}$ (which is the special case with $m=\tau$). Thus, as in Theorem \ref{thm:CRLB}, $\lim_{k\uparrow \infty}MSE_{k,m}^{\star}\geq \frac{1}{J_{\infty}}+\Ocal(1/\tau)$. 

Fix a time $t> \tau$ and let $m(t),k(t)$ be such that $m(t)+k(t)\tau=t$ for some $k(t)$; in which case $k(t)=t-m(t)$. Then, 
$$MSE_t^{\star}= MSE_{k(t),m(t)}^{\star}
\geq \min_{m\in[\tau]}MSE_{k(t),m}^{\star}.$$ The first equality is standard and follows from the fact that the conditional expectations (which achieve the optimal MSE) are identical in the sequential (period-by-period) or the batched setting. In turn, 
    $$\liminf_{t\uparrow \infty} MSE_t^{\star}\geq \lim_{k\uparrow \infty}\min_{m\in[\tau]}MSE^{\star}_{k,m}\geq \frac{1}{J_{\infty}}- \Ocal\lp \frac{1}{\tau \bar{J}_{\infty}}\rp.$$
\end{proof}

For both the centered and non-centered filters setting $\tau=s_N$ in Theorem \ref{thm:CRLB} and using $1\lesssim s_N\lesssim\sqrt{N}$ (balanced regime) gives us a lower bound that is---up to $1/\sqrt{N}$---equal to $1/\bar{J}_{\infty}$. The upper bounds, however, will be different and explain the stronger result we have for the centered filter. We now turn to proving Theorem \ref{thm:CRLB}.

First, the following lemma collects several simple algebraic facts about $\Sigma(\calW)$, ${\bf e}'I(\mathbb{V}){\bf e}$ and $\Ex[X_t^k]$ that are frequently used in our proofs.

\begin{lemma} For \label{lem:algebra} 
 $\tau=\Ocal(N)$ 
$$\frac{1-\gamma^{2\tau}}{1-\gamma^2} = \Theta(\tau), ~~\Sigma^{-1}(\calW) = N\frac{1-\gamma^2}{1-\gamma^{2\tau}} = \Theta\lp\frac{N}{\tau}\rp,~ \Sigma(\calW) = \Theta\lp\frac{\tau}{N}\rp.$$
$${\bf e}'I(\mathbb{V} ){\bf e} = \frac{\gamma^2}{s_N^2\gamma^{2\tau}}I(v)\frac{1-\gamma^{2\tau}}{1-\gamma^2}=\Theta\lp\frac{\tau}{s_N^2}\rp,$$

Finally, with $X_0=0$ and for all integer $k$ and all $t$ (including steady-state)  $$\Ex[X_t^{2k}]=\Ocal(1).$$

\end{lemma}

\begin{proofof}{Theorem \ref{thm:CRLB}} 
Let us decompose $\Vcal_k=\mathbb{V}_k^-+\mathbb{W}_k+{\bf e}_ks_N v_{\tau k}$ 
where ${\bf e}_k$ is the unit vector that has $1$ in the $\tau$'s entry and $0$ elsewhere, and 
\begin{align*}
\mathbb{V}_k^-&=\left(\gamma^{\tau k-t}s_N v_t,~ t\in [\tau(k-1)+1,\ldots,\tau k-1]\right),~~ 
\mathbb{V}_k=[\mathbb{V}_k^{-},s_N v_{\tau k}]\\ 
\mathbb{W}_k& = \left(-\frac{1}{\sqrt{N}} \sum_{s=1}^{\tau k-t}\gamma^{s-1} w_{\tau k-s}, ~ t\in [\tau(k-1)+1,\ldots,\tau k-1]\right).
\numberthis \label{eq:v_kw_k}
\end{align*}

Both vectors, $\mathbb{V}_k^-$ and $\mathbb{W}_k$, are of dimension $\tau-1$. Because the sequence is i.i.d we drop the index $k$. 

The vector $\mathbb{V}^-$ has independent components but the components of $\mathbb{W}$
are dependent;  $\mathbb{V}^-$ and $\mathbb{W}$ are independent of each other. The information matrix of $\calV$ is then the block matrix
$$I(\calV) = \left[\begin{array}{cc}  I(\mathbb{V}^{-}+\mathbb{W}) & 0 \\
0 & I(v)/s_N^2.
\end{array}\right].$$ 

The Fisher information matrix is the covariance matrix of the individual score functions; it is trivially positive semi definite. The Fisher information matrix of $\Vbb^{-}$, denoted by $I(\mathbb{V}^-)$ (recall that the dependence on $k$ is dropped due to stationarity) is of the following form
\begin{align}
    I(\mathbb{V}^-) = \begin{bmatrix}
        \frac{I(v)}{\gamma^{2(\tau -1)}s_N^2} & 0 &0 &\dots\\
        0& \frac{I(v)}{\gamma^{2(\tau -2)}s_N^2 }& 0 & \dots\\
        \vdots & \vdots & \vdots&
    \end{bmatrix}
\end{align}
is diagonal, with the $i^{th}$ element equal to $\frac{I(v)}{\gamma^{2(\tau  - i)}s_N^2}$; it is invertible and positive definite. The vector $\mathbb{W}_k$ is an invertible linear transformation $A$ of the vector ${\bf w}_k = (w_{\tau(k-1)+1},\ldots ,w_{\tau k -1})$. In other words, $\Wbb 
= A\, \mathbf{w}$ where $A\in \mathbb{R}^{(\tau - 1)\times(\tau - 1)}$ is the \emph{invertible} upper-triangular matrix 
\[
A(\gamma,\tau)
\;:=\;-\frac{1}{\sqrt{N}}
\begin{bmatrix}
\gamma^{\tau-2} & \gamma^{\tau-3} & \cdots & \gamma & 1\\
0 & \gamma^{\tau-3} & \cdots & \gamma & 1\\
\vdots & \ddots & \ddots & \vdots & \vdots\\
0 & \cdots & 0 & \gamma & 1\\
0 & \cdots & 0 & 0 & 1
\end{bmatrix}
\in \mathbb{R}^{(\tau-1)\times(\tau-1)}.
\]

The vector-valued random variable
$\mathbb{V}^{-} + \mathbb{W}$ is itself a linear  transformation of the $2(\tau-1)$-dimensional vector $({\bf v}_k,{\bf w}_k)$ where ${\bf w}_k$ is as defined above and ${\bf v}= (\gamma^{\tau k-t} s_N v_{t}, t\in [\tau(k-1),\ldots,\tau k-1])$; that is 
$$\mathbb{V}^{-}_k+\mathbb{W}_k= B({\bf v},{\bf w}), \mbox{ where } B:=[\mathbb{I}_{\tau-1}, A],$$ 
with $\mathbb{I}_{\tau-1}$ the $(\tau-1)\times (\tau-1)$ identity matrix. The matrix $B$ has full row rank. 

For matrices $M_1,M_2$, we say $M_1\geq M_2$ if $M_1-M_2$ is positive semi-definite and $M_1>M_2$ if $M_1-M_2$ is positive definite. A version of the Fisher information (convolution) inequality \cite[Corollary 1]{zamir} (see also \cite{bercher2002matrixfisher}) states that 
$$I(\mathbb{V}^{-}+\mathbb{W}) \preceq (BI^{-1}({\bf v},{\bf w})B^T)^{-1},$$ 
because all random variables in $({\bf v},{\bf w})$ are independent of each other, $I({\bf v},{\bf w})$ is a diagonal matrix 
with the first $\tau-1$ entries equal to $I(v)/(\gamma^{2(\tau k-t)}s_N^2)$ for the $\tau k-t$ entry and the latter $\tau-1$ equal to $N/I(w)$. 

We thus have that 
\begin{align*} 
&BI^{-1}({\bf v},{\bf w})B^T 
= I^{-1}({\bf v}) + AI^{-1}({\bf w})A^T
= I^{-1}(\mathbb{V}^-) + AI^{-1}({\bf w})A^T.
\end{align*} 

Because $A$ is invertible  and $I({\bf w})$ has strictly positive diagonal entries, the matrix $AI^{-1}({\bf w})A^T$ is positive definite (PD). We conclude that  
\begin{align*} 
I(\mathbb{V}^-+\mathbb{W})
&\preceq (BI^{-1}({\bf v},{\bf w})B^T)^{-1} \\
&= (I^{-1}(\mathbb{V}^-) + AI^{-1}({\bf w})A^T)^{-1} \\
&\preceq I(\mathbb{V}^-).
\end{align*} 

In the last inequality we used a fact for PD matrices (that for  $M_1,M_2$  PD matrices it holds $(M_1+M_2)^{-1} \preceq M_1^{-1}, (M_1+M_2)^{-1}\preceq M_2^{-1}$).

We arrive at 
$$I(\calV) = \left[\begin{array}{cc}  I(\mathbb{V}^-+\mathbb{W}) & 0 \\
0 & I(v)/s_N^2
\end{array}\right]\preceq\left[\begin{array}{cc}  I(\mathbb{V}^-) & 0 \\
0 & I(v)/s_N^2
\end{array}\right]$$ 
and, by the definition of PSD, that 
\begin{equation} {\bf e}'I(\calV){\bf e}\leq {\bf e}'I(\mathbb{V}){\bf e};\label{eq:VtoV}\end{equation} 
here $\mathbb{V}=[\mathbb{V}^-,s_N v_{\tau}]$. Recalling \cref{eq:tichavsky} we then have  
$$      J_{k+1}  =  {\bf e}'I(\calV){\bf e}+\frac{J_kI(\calW)}{J_k+\gamma^{2\tau}I(\calW)}, $$ we have $$J_{k+1} \leq {\bf e}'I(\mathbb{V}){\bf e}+ \frac{J_kI(\calW)}{J_k+\gamma^{2\tau}I(\calW)}.$$

Let $\tilde{J}$ solve the recursion 
$$\tilde{J}_{k+1} = {\bf e}'I(\mathbb{V}){\bf e}+ \frac{\tilde{J}_kI(\calW)}{\tilde{J}_k+\gamma^{2\tau}I(\calW)}.$$
Using \eqref{eq:VtoV}, we have that $\tilde{J}_0=J_0$, implies, by induction, $J_k\leq \tilde{J}_k$ for all $k$. This is because, if $J_k\leq \tilde{J}_k$, then 
$$J_{k+1}- \tilde{J}_{k+1} =  {\bf e}'I(\calV){\bf e}-{\bf e}'I(\mathbb{V}){\bf e} + 
\frac{\gamma^{2\tau}I^2(\calW)(J_k-\tilde{J}_k)}{(J_k+\gamma^{2\tau}I(\calW))(\tilde{J}_k+\gamma^{2\tau}I(\calW))}\leq 0.$$ 

In turn, $MSE^{\star}_k \geq \frac{1}{J_k}$ implies now
$$MSE^{\star}_k\geq \frac{1}{\tilde{J}_k}.$$

Let $\bar{J}_{\infty}$ be, as in the statement of the theorem, the stationary solution of the recursion 
$$\bar{J}_{k+1} = {\bf e}'I(\mathbb{V}){\bf e} + \frac{\bar{J}_k\Sigma^{-1}(\calW)}{\bar{J_k}+\gamma^{2\tau} \Sigma^{-1}(\calW)}.$$ 

We will prove that 
\begin{equation}  \label{eq:tildeJJgap} |\tilde{J}_{\infty}-\bar J_{\infty}|=\Theta\left(\frac{1}{\tau}\frac{\sqrt{N}}{s_N}\right),~~\bar{J}_{\infty}=\Theta\left(\max\left\{\frac{\tau}{s_N^2}+\frac{\sqrt{N}}{s_N},1\right\}\right),\end{equation} which guarantees that for $\tau\gg 1$, $|\tilde{J}_{\infty}-\bar{J}_{\infty}| = o(\bar{J}_{\infty})$. In turn, as stated in Theorem\ref{thm:CRLB},  \begin{align*}
    \left|\frac{1}{\tilde J_{\infty}}-\frac{1}{\bar J_{\infty}}\right| & = \frac{|\tilde{J}_{\infty}-\bar J_{\infty}|}{\tilde{J}_{\infty}\bar J_{\infty}}= \Ocal\lp \frac{1}{\tau \bar{J}_{\infty}} \rp.
\end{align*} for $tau\gg 1$, 
and 
$$MSE^{\star}\geq \frac{1}{J_k}\geq \frac{1}{\tilde J_k}\geq \frac{1}{\bar{J}_k} -\Ocal\left(\frac{1}{\tau \bar{J}_{\infty}}\right).$$

\textbf{Proof of \cref{eq:tildeJJgap}:} The remainder of this section is dedicated to the proof of \eqref{eq:tildeJJgap}. Below $\calW$ is as defined in \ref{sec:GFcenteredUB} and depends on $\tau$. 
\begin{lemma}[Fisher information CLT for $\calW$] 
$$1\leq I(\calW)\Sigma(\calW)= 1 + \Ocal\left(\frac{1}{\tau}\right).$$
\label{lem:CLT_gn}
\end{lemma} 
\noindent Recall that all lemma proofs appear in the Appendix \ref{sec:auxproofs}.

Let 
\[
\delta=\delta_{N,\tau} = I(\calW)\Sigma(\calW)-1.
\] 
Notice that $\delta>0$ for all distributions except for Gaussian where $\delta=0$. We write $I(\calW) = \Sigma^{-1}(\calW)(1+\delta)$ and note that, by Lemma \ref{lem:CLT_gn}, $\delta=\Ocal(1/\tau)$.

We are comparing the recursion
\begin{align}
\bar{J}_{k+1} = {\bf e}'I(\mathbb{V}){\bf e} + \frac{\bar{J}_k\Sigma^{-1}(\calW)}{\bar{J_k}+\gamma^{2\tau}\Sigma^{-1}(\calW) },\numberthis\label{eq:approximate-recursion}
\end{align} to the recursion 
\begin{align}\tilde{J}_{k+1} & = {\bf e
}'I(\mathbb{V}){\bf e}+ \frac{\tilde{J}_kI(\calW)}{\tilde{J}_k+\gamma^{2\tau}I(\calW)}\nonumber\\ & ={\bf e}'I(\mathbb{V}){\bf e}+
\frac{\tilde{J}_k\Sigma^{-1}(\calW)(1+\delta)}{\tilde{J}_k+\gamma^{2\tau}\Sigma^{-1}(\calW)(1+\delta)}.\label{eq:tildeJrecursion}
\end{align}

To bound $|\bar J_{\infty}-\tilde J_{\infty}|$, we solve for the steady-state in both equations \eqref{eq:approximate-recursion} and \eqref{eq:tildeJrecursion}. \footnote{We use throughout the standard fact that a sequence of the form $a_{k+1} = b + \frac{a_k c}{a_k+\alpha c},$ for $b,c\geq 0$ and $\alpha \in(0,1)$ is convergent to the unique fixed point of the recursion.} These result in quadratic equations with the solutions (only the positive roots are relevant)
\begin{align*}
    \bar J & = \frac{{\bf e}'I(\mathbb{V}){\bf e}+(1-\gamma^{2\tau})\Sigma^{-1}(\calW)+\sqrt{\bigl({\bf e}'I(\mathbb{V}){\bf e}+(1-\gamma^{2\tau})\Sigma^{-1}(\calW) \bigr)^2+4{\bf e}'I(\mathbb{V}){\bf e}\gamma^{2\tau}\Sigma^{-1}(\calW)}}{2}\numberthis\label{eq:rec-barJinf}\\
    & =  \frac{{\bf e}'I(\mathbb{V}){\bf e}+(1-\gamma^{2\tau})\Sigma^{-1}(\calW)+\sqrt{A}}{2}, \\
    \tilde J & = \frac{{\bf e}'I(\mathbb{V}){\bf e}+(1-\gamma^{2\tau})(1+\delta)\Sigma^{-1}(\calW)+\sqrt{\bigl({\bf e}'I(\mathbb{V}){\bf e}+(1+\delta)(1-\gamma^{2\tau})\Sigma^{-1}(\calW)\bigr)^2+4{\bf e}'I(\mathbb{V}){\bf e}\gamma^{2\tau}(1+\delta)\Sigma^{-1}(\calW)}}{2}\\
    & = \frac{{\bf e}'I(\mathbb{V}){\bf e}+(1-\gamma^{2\tau})(1+\delta)\Sigma^{-1}(\calW)+\sqrt{B}}{2},
\end{align*} 

Subtracting, we have
\begin{align}
    |\bar J-\tilde J| & \leq \frac{\lvert (1-\gamma^{2\tau})\delta \Sigma^{-1}(\calW)\rvert}{2}+\left\lvert\frac{\sqrt{A}-\sqrt{B}}{2}\right\rvert.\label{eq:interim} 
\end{align} 

From Lemma \ref{lem:algebra} we have that $(1-\gamma
^{2\tau})\Sigma^{-1}(\calW)=N(1-\gamma^2) = \Theta(1),$ so that the first term in \eqref{eq:interim} has

\begin{align}
    \frac{1}{2}\left\lvert (1-\gamma^{2\tau})\delta\Sigma^{-1}(\calW)\right\rvert\overset{(i)}{=} \Ocal(\delta)\numberthis.\label{eq:J-diffterm1}
\end{align}

Multiplying and dividing the second term of \eqref{eq:interim} by $\sqrt{A}+\sqrt{B}$ we have 
\begin{align*}
   |\sqrt{A}-\sqrt{B}|= \left|\frac{A-B}{\sqrt{A}+\sqrt{B}} \right| \leq \frac{|C^2-D^2|}{\sqrt{A}+\sqrt{B}}+
   \frac{4{\bf e}'I(\mathbb{V}){\bf e}\gamma^{2\tau}\delta \Sigma^{-1}(\calW)}{\sqrt{A}+\sqrt{B}}   ,\numberthis\label{eq:J-diffterm2}
\end{align*}
where
$$C := {\bf e}'I(\mathbb{V}){\bf e}+(1-\gamma^{2\tau})\Sigma^{-1}(\calW)\leq \sqrt{A} , ~~ D:= {\bf e}'I(\mathbb{V}){\bf e}+(1+\delta)(1-\gamma^{2\tau})\Sigma^{-1}(\calW)\leq \sqrt{B}.$$
Then, 
\begin{align*}
    |C^2-D^2| & = \left|(C-D)(C+D)\right|\\
    & = (1-\gamma^{2\tau})\delta \Sigma^{-1}(\calW)(C+D)\leq (1-\gamma^{2\tau})\delta \Sigma^{-1}(\calW)(\sqrt{A}+\sqrt{B}).
\end{align*}
Using \eqref{eq:J-diffterm1} we then get
\begin{align*}
    \frac{|C^2-D^2|}{\sqrt{A}+\sqrt{B}}\leq (1-\gamma^{2\tau})\delta \Sigma^{-1}(\calW)= \Ocal(\delta)   .\numberthis\label{eq:J-diffterm3}
\end{align*}
The second term on the right-hand side of \eqref{eq:J-diffterm2} has
\[
\frac{4{\bf e}'I(\mathbb{V}){\bf e}\gamma^{2\tau}\delta \Sigma^{-1}(\calW)}{\sqrt{A}+\sqrt{B}}\leq\frac{4{\bf e}'I(\mathbb{V}){\bf e}\gamma^{2\tau}\delta \Sigma^{-1}(\calW)}{\sqrt{A}} \leq \delta  \sqrt{4{\bf e}'I(\mathbb{V}){\bf e}\gamma^{2\tau}\Sigma^{-1}(\calW)},
\]
where we trivially lower bound $A$ with $4{\bf e}'I(\mathbb{V}){\bf e}\gamma^{2\tau}\Sigma^{-1}(\calW)$. 

Substituting the bound ${\bf e}'I(\Vbb){\bf e}=\Ocal(\tau/s_N^2)$ and the value of $\Sigma(\calW) = \Var(\calW)= (1-\gamma^{2\tau})/(N(1-\gamma^2))$---both from Lemma\ref{lem:algebra}---we get \begin{align*}
    \delta  \sqrt{4{\bf e}'I(\mathbb{V}){\bf e}\gamma^{2\tau}\Sigma^{-1}(\calW)} = \Ocal\lp \delta  \frac{\sqrt{N}}{s_N}\rp.\numberthis\label{eq:Jdiff-bounds}
\end{align*}
We thus have the first required bound of \eqref{eq:tildeJJgap}. For the second bound there---that $ \bar{J}=\Theta\left(\max\left\{\frac{\tau}{s_N^2}+\frac{\sqrt{N}}{s_N},1\right\}\right)$---recall that 
\begin{align*}
    \bar J & = \frac{{\bf e}'I(\mathbb{V}){\bf e}+\Sigma^{-1}(\calW)(1-\gamma^{2\tau})+\sqrt{\bigl({\bf e}'I(\mathbb{V}){\bf e}+\Sigma^{-1}(\calW)(1-\gamma^{2\tau})\bigr)^2+4{\bf e}'I(\mathbb{V}){\bf e}\gamma^{2\tau}\Sigma^{-1}(\calW)}}{2},
\end{align*}  
and that, by Lemma \ref{lem:algebra},
\begin{align*}
    {\bf e}'I(\mathbb{V}){\bf e}=\Theta\left(\frac{\tau}{s_N^2}\right);\,  (1-\gamma^{2\tau} )\Sigma^{-1}(\calW) = (1-\gamma^2)N = 2 + \ocal(1) = \Theta \left(1\right).\numberthis\label{eq:facts}
\end{align*}
Therefore, 
$${\bf e}'I(\mathbb{V}){\bf e}+\Sigma^{-1}(\calW)(1-\gamma^{2\tau})
=\Theta\left(1+\frac{\tau}{s_N^2}\right).$$ 
Using the explicit expressions in Lemma \ref{lem:algebra} once again we have, for $\tau=\Ocal(N)$, that 
$$4{\bf e}'I(\mathbb{V}){\bf e}\gamma^{2\tau}\Sigma^{-1}(\calW)=4\gamma^2\frac{N}{s_N^2}I(v)=\Theta\lp\frac{N}{s_N^2}\rp,$$ 
so that overall 
$${\bf e}'I(\mathbb{V}){\bf e}+\Sigma^{-1}(\calW)(1-\gamma^{2\tau})+\sqrt{\bigl({\bf e}'I(\mathbb{V}){\bf e}+\Sigma^{-1}(\calW)(1-\gamma^{2\tau})\bigr)^2+4{\bf e}'I(\mathbb{V}){\bf e}\gamma^{2\tau}\Sigma^{-1}(\calW)}=\Theta\left(1+\frac{\tau}{s_N^2}+\frac{\sqrt{N}}{s_N}\right),$$ 
as stated. Finally, we note that if $w_t$ is Gaussian, $\delta\equiv 0$ and $I(\calW)=1/\Var(\calW)$ so we have $MSE^{\star}\geq 1/\bar{ J}_{\infty}$.
\end{proofof}

\section{The Goggin filter upper bound\label{sec:upperbound}}

This section will be dedicated to proving the upper bounds on the bias and variance of the Goggin filter. In particular, we will prove the following theorem:
\begin{theorem}\label{thm:upperbound} Suppose $1\ll s_N\lesssim \sqrt{N}$. Let $\hat{x}_t$ be the (random) estimator produced by the Goggin filter and $\bar J_{\infty}$ be the solution to the equation 
\begin{align*}
\bar{J}_{\infty} = {\bf e}'I(\mathbb{V}){\bf e} + \frac{\bar{J}_{\infty}\Sigma^{-1}(\calW)}{\bar{J}_{\infty}+\gamma^{2\tau}\Sigma^{-1}(\calW) },
\end{align*}  with $\tau=s_N$.  Then, in stationarity, 
\begin{align} 
\Ex[\hat{x}_t-X_t] &= \Ocal\left(\frac{1}{s_N}\right) \tag{Bias} \\
\Ex[(\hat{x}_t-X_t)^2]& = \frac{1}{\bar J_{\infty}}+ \Ocal\left(\frac{1}{s_N}\right).\tag{Variance} \end{align} 
\end{theorem}

\vspace*{0.5cm} 
\noindent \textbf{Proof of Theorem \ref{thm:main}:} Corollary \ref{thm:CRLB}, and Theorem \ref{thm:upperbound} prove Theorem \ref{thm:main}. \hfill \qed

We now prove Theorem \ref{thm:upperbound}. Recall the Goggin filter and associated notations in Definition  \ref{def:goggin-filter}, and let us consider the Bias and Variance separately starting with the former. 

\textbf{Bias:} 
Let $\tilde{x}_0=X_0=0$ 
\begin{align*}
    \tilde{x}_t & = (1-K_t I(v))\gamma \tilde{x}_{t-1} +K_t(I(v)X_t + s_N\phi(v_t)).
\end{align*} Simple induction shows that 
\begin{align*}
    \expec[\tilde{x}_t-X_t] = 0.
\end{align*}

Recall the transformed observation in time $t$
\begin{align*}
Z_{t}&:=s_N \phi \left( \frac{1}{s_N}Y_t\right) = s_N \phi \left( \frac{1}{s_N}X_t + v_t\right).
\end{align*} 
and the update equation of the Goggin filter as given by
\begin{align*}
    \hat{x}_t =& \left(1-K_tI(v)\right)\gamma \hat{x}_{t-1} + K_tZ_t,
\end{align*} with $K_t$ as specified in \eqref{eq:KPdefin}. We will show that
\begin{align*}
    \Ex[\hat{x}_t-\tilde x_t]=\Ocal(1/s_N),\numberthis\label{eq:bias-bound}
\end{align*} which implies that 
$$\Ex[\hat{x}_t-X_t]= \Ex[\hat{x}_t-\tilde{x}_t]+ \Ex[\tilde{x}_t-X_t] = \Ocal(1/s_N),$$ as stated. 

To that end, write 
\begin{equation}  
Z_{t}=s_N\phi\left(\frac{1}{s_N}X_t + v_t\right)=
s_N\phi(v_t)+\Ex[\phi'(v_t)]X_t+\remainder_t
\label{eq:Zexpansion}\end{equation}
where 
\begin{align*} 
\remainder_t  &= s_N\phi\left(\frac{1}{s_N}X_t + v_t\right) - 
s_N\phi(v_t)-\Ex[\phi'(v_t)]X_t.\end{align*} 
The difference $\hat x_t-\tilde x_t$ is then written as
\begin{align*}
\hat{x}_t - \tilde{x}_t = (1-K_t I(v))\gamma (\hat{x}_{t-1}-\tilde{x}_{t-1})+K_t\remainder_t.\numberthis\label{eq:hatminustilde}    
\end{align*}

Taylor  expansion on $\phi$ gives  
\begin{align*} \remainder_t & = X_t\lp \phi'(v_t)-\Ex[\phi'(v_t)
]\rp + \frac{1}{2}\phi''(\varphi_t) \frac{1}{s_N}X_t^2, \numberthis\label{eq:R_t}
\end{align*} where $\varphi_t\in [\min\{v_t,v_t+X_t/s_N\},\max\{v_t,v_t+X_t/s_N\}]$.
Since $|\phi''(\varphi_t)|\leq \|\phi''\|_\infty<\infty$ by Assumption \ref{asum:primitives-observation}, we have 
$$\Ex[\remainder_t]=\Ocal\left(\frac{1}{s_N} \Ex[X_t^2]\right).$$ 

The following lemma will support bounding the remainder term and its interaction with other terms in the recursion \cref{eq:hatminustilde}.
\begin{lemma} Suppose that $1\ll  s_N\lesssim \sqrt{N}$. In stationarity, $P_t,K_t$ as defined in \eqref{eq:KPdefin} satisfy $$P_{t}=\Theta\lp \frac{s_N}{\sqrt{N}}\rp, \mbox{ and } 
K_t =\Theta\left(\frac{1}{s_N\sqrt{N}}\right).$$ Also, with $X_0=0$, we have 
$$\Ex[X_t^2]=\Ocal(1), \mbox{ for all t}.$$\label{lem:KFPXg}
\end{lemma} 

Substituting the bounds from Lemma \ref{lem:KFPXg} we get $|\Ex[K_t\varrho_t] |= \Ocal(1/(s_N^2\sqrt{N}))$ and, in \cref{eq:hatminustilde}, that 
$$\Ex[\hx_t-\tilde{x}_{t}] =  (1-K_tI(v))\gamma\Ex[\hx_{t-1}-\tilde{x}_{t-1}] \pm  \Ocal\left(\frac{1}{s_N^2\sqrt{N}}\right).$$

Let $f(t):=|\expec[\hat x_t-\tilde x_t]|$. Using Lemma \ref{lem:KFPXg}, we have, for a suitable constant $\Gamma$, that 
\begin{align*}
    f(t) \leq \lp 1-\frac{1}{\Gamma}\frac{I(v)}{s_N\sqrt{N}} \rp f(t-1) + \Gamma \frac{1}{s_N^2\sqrt{N}}.
\end{align*}
Iterating this inequality we get, for a re-defined $\Gamma$, 
\begin{align*}
    f(t) \leq \Gamma\frac{1}{s_N^2\sqrt{N}}\sum_{i=0}^{t-1}\lp1-\frac{I(v)}{s_N\sqrt{N}} \rp^i + \Gamma f(0)\lp1-\frac{I(v)}{s_N\sqrt{N}} \rp^{t}
\end{align*}
Completing the summation we get that $f(t) = \Ocal\lp\frac{1}{s_N}\rp$ for $f(0) =\Ocal(1)$ and all $t$. In particular, in stationarity 
$$|\Ex[\hat{x}_t-\tilde{x}_t]|=\Ocal(1/s_N).$$ which proves the bound on the bias \cref{eq:bias-bound}. \whiteqed

\textbf{Variance:} Recall that
$$\hx_t = (1-K_t I(v))\gamma\hx_{t-1}+K_t Z_t.$$ 
Let $\hat P_t := \Var(X_t-\hat x_t)$, and $J_t^{GF}:=1/P_t$ where $P_t$ as defined in \eqref{eq:KPdefin}. Recall also the definition of $\bar J_{\infty}$ from Theorem \ref{thm:CRLB}. We must prove that $\hat P_t-\frac{1}{\bar{J}_{\infty}} = \Ocal(1/s_N)$. Using the triangle inequality we write
\begin{align}
    |\hat P_\infty-1/\bar{J}_{\infty}| \leq  |\hat P_\infty-P_\infty|+|P_\infty-1/J_\infty^{GF}|+|1/J_\infty^{GF}-1/\bar J_\infty|,\label{eq:variancebound}
\end{align} where $P_{\infty}$ is the stationary point of $P_t$, and $J^{GF}_{\infty} = 1/P_{\infty}$. We will bound each term separately. The second term $|P_{\infty}- 1/J^{GF}_{\infty}|$
is then equal to $0$, and we will bound each of the other two terms. 

\noindent \textbf{Step I $|\hat P_\infty-P_\infty|$:}  Recalling the expansion of $Z_t$ in \eqref{eq:Zexpansion} we have 
\begin{align*} 
X_t-\hx_t& = X_t - (1-K_t I(v))\gamma\hx_{t-1} - K_tZ_t \\
& = (1-K_tI(v)) (X_t-\gamma\hx_{t-1})-K_ts_N\phi(v_t) - K_t\remainder_t\\ & = 
(1-K_tI(v))\lp\gamma(X_{t-1}-\hx_{t-1})+\frac{1}{\sqrt{N}}w_{t-1}\rp - K_ts_N\phi(v_t)-K_t\remainder_t
\end{align*} 
where in the third equality we use $X_t = \gamma X_{t-1}+\frac{1}{\sqrt{N}}w_{t-1}$.
Using the notation 
$$R:=s_N^2I(v), ~~ Q = 1/N,$$ we have 
\begin{align} \hat P_t & = \Var(X_t-\hat{x}_t) \nonumber\\
& = (1-K_tI(v))^2 [\gamma^{2}\hat P_{t-1} + Q] 
+K_t^2 R + K_t^2\Var(\remainder_t) + 2Cov(K_t\remainder_t,A_t), \label{eq:decompVar}
\end{align}
where $$A_t:=  
(1-K_tI(v))[\gamma (X_{t-1}-\hx_{t-1})+\frac{1}{\sqrt{N}}w_{t-1}] - K_ts_N\phi(v_t).$$
The first two terms in \eqref{eq:decompVar} correspond to the Joseph form of variance update equation of a standard Kalman-Filter.
The last two terms are the contribution of the approximation error to the recursive update. We bound these terms using the following Lemma \ref{lem:remainderterms1}.  
\begin{lemma} 
$$K_t^2\Var(\remainder_t) + 2|Cov(K_t\remainder_t,A_t)| = \Ocal\left(\frac{1}{s_N^2\sqrt{N}}(1+\hat{P}_{t-1})\right).$$ \label{lem:remainderterms1}
\end{lemma}
Using Lemma \ref{lem:remainderterms1} to bound the corresponding terms in \cref{eq:decompVar}, we are left with
\begin{align}\label{eq:hPrecursion}
    \hat{P}_t = (1-K_tI(v))^2 (\gamma^{2} \hat{P}_{t-1}+ Q) +K_t^2R\pm \Ocal\lp\frac{1}{s_N^2\sqrt{N}}(1+\hat{P}_{t-1})\rp.
\end{align}

Recalling that $K_t= \Theta(1/(s_N\sqrt{N}))$ we have the existence of a constant $\eta>0$ such that $(1-K_tI(v))^2 \leq \left(1-\eta \frac{1}{s_N\sqrt{N}}\right)$ for all $N$ sufficiently large.  Using $Q = 1/N$ and $K_t^2 R = \Ocal(1/N)$ we then have 
$$\hat{P}_t \leq \lp 1-\eta\frac{1}{s_N\sqrt{N}}\rp \hat{P}_{t-1}+\Ocal\lp\frac{1}{N}+\frac{1}{s_N^2\sqrt{N}}.\rp$$ 

In stationarity, using the assumption $s_N\lesssim\sqrt{N}$,
$$\hat{P}_t =\Ocal\lp\max\left\{\frac{1}{s_N},\frac{s_N}{\sqrt{N}}\right\}\rp=\Ocal(1).$$ 

Comparing now the recursion, for $\hat{P}_t$ in \eqref{eq:hPrecursion} with that for $P_t$  
\begin{align*}
    P_t = (1-K_tI(v))^2 (\gamma^{2} P_{t-1}+ Q) +K_t^2R
\end{align*}
we have
\begin{align*}
    P_t-\hat{P}_t & = (1-K_tI(v))^2(\gamma^{2}(P_{t-1}-\hat{P}_{t-1})) \pm \Ocal\lp\frac{1}{\sqrt{N}s_N^2}(1+\hat{P}_{t-1})\rp \\ & = (1-K_tI(v))^2(\gamma^{2}(P_{t-1}-\hat{P}_{t-1})) \pm  \Ocal\lp\frac{1}{\sqrt{N}s_N^2}\rp
\end{align*}
Recalling again that, in stationarity, $K_t=\Theta(1/(s_N\sqrt{N}))$ and setting $f(t):=|P_t-\hat{P}_t|$ we have for a constant $\Gamma>0$
\begin{align*}
     f(t) \leq  \lp 1-\frac{1}{\Gamma} \frac{I(v)}{s_N\sqrt{N}}\rp^2f(t-1) + \Gamma \lp\frac{1}{\sqrt{N}s_N^2}\rp.
\end{align*}

Iterating this recursion we get for $f(0)=\Ocal(1)$ and then in stationarity 
\begin{align*}
    f(t) = |P_t-\hat{P}_t| & = \Ocal\lp\frac{1/(s_N^2\sqrt{N})}{1/(s_N\sqrt{N})}\rp = \Ocal\lp\frac{1}{s_N}\rp.
\end{align*}  
\whiteqed
\noindent \textbf{Step II $|1/J_\infty^{GF}-1/\bar J_\infty|$:} We bound this term by direct calculation of these two quantities. Recall that $P_t$ follows the recursion
$$ P_t = \frac{R(\gamma^2P_{t-1}+Q)}{(I(v))^2 (\gamma^2 P_{t-1}+Q)+R  }.$$
$\JGF_t= 1/P_t$ then follows the recursion (recall that $Q=1/N$ and $R=s_N^2 I(v)$)  
\begin{align*}
    \JGF_{t} = \frac{1}{s_N^2}I(v) + \frac{\JGF_{t-1}N}{\JGF_{t-1} + \gamma^2N}.
\end{align*}
This recursion's stationary value, $\JGF_{\infty}$ is the positive solution to the equation 
$$ 
(\JGF)^2 -\left((1-\gamma^{2}){N}+\frac{1}{s_N^2}I(v)\right)\JGF -\frac{\gamma^2 N}{s_N^2}I(v)=0,$$ given by 
\[ 
\JGF_{\infty} =\frac{1}{2}\lb\left((1-\gamma^{2}){N}+\frac{1}{s_N^2}I(v)\right)+\sqrt{\left((1-\gamma^{2}){N}+\frac{1}{s_N^2}I(v)\right)^2 + \frac{4\gamma^2 N}{s_N^2}I(v)}\rb.
\]

We compare this stationary value with that for $\bar{J}_\infty$ in Theorem \ref{thm:CRLB} which is the positive solution to the equation  
\begin{align*}
\bar{J}_\infty = {\bf e}'I(\mathbb{V}){\bf e} + \frac{\bar{J}_\infty\Sigma^{-1}(\calW)}{\bar{J}_\infty+\gamma^{2\tau}\Sigma^{-1}(\calW) },
\end{align*} as given by 
$$\bar{J}_{\infty} = \frac{1}{2}\left[(1-\gamma^{2\tau})\Sigma^{-1}(\calW)+\bar{I} + \sqrt{ ((1-\gamma^{2\tau})\Sigma^{-1}(\calW)+\bar{I})^2 + 4 \gamma^{2\tau} \Sigma^{-1}(\calW)\bar{I}}\right].$$ 

In $\bar{J}_{\infty}$, compared to $J^{GF}_{\infty}$, $\gamma^2$ is replaced with $\gamma^{2\tau}$ and $N$ is replaced with $\Sigma^{-1}(\calW)$. We will establish that, with the choice $\tau=s_N$, 
\[
|\bar{J}_{\infty}-\JGF_{\infty}| = \Ocal\left(\frac{1}{s_N}\right).\numberthis\label{eq:Jgap}
\] 

Henceforth, we drop the subscript $\infty$ for convenience. We note 
$$|\bar{J}-\JGF| \leq A + B,$$
where
\begin{align}\label{eq:A}
  A & = \frac{1}{2}\left|(1-\gamma^{2\tau})\Sigma^{-1}(\calW) - (1-\gamma^2)N +\bar{I}-\frac{1}{s_N^2}I(v)\right| \text{ and, }\\
  B & = \frac{1}{2}\left|\underbrace{\sqrt{ ((1-\gamma^{2\tau})\Sigma^{-1}(\calW)+\bar{I})^2 + 4 \gamma^{2\tau} \Sigma^{-1}(\calW)\bar{I}}}_{B_1}-\underbrace{\sqrt{ \left((1-\gamma^{2}){N} +\frac{1}{s_N^2}I(v)\right)^2 + \frac{4 \gamma^{2}N}{s_N^2}I(v)}}_{B_2}\right|.\nonumber
\end{align}

Recalling that $\Sigma^{-1}(\calW)= N\frac{1-\gamma^2}{1-\gamma^{2\tau}}$ as well as $\bar{I} = \Theta(\nicefrac{\tau}{s_N^2})$ (see Lemma \ref{lem:algebra}),  we have for $\tau=\Omega(1)$ that 
\begin{align*}
  A & = \Ocal\lp\frac{\tau}{s_N^2}\rp, \text{ and } B = \frac{1}{2}\left|\frac{B_1^2-B_2^2}{B_1+B_2}\right|.
\end{align*} 

Observe that 
\begin{align*}
    B_1^2-B_2^2 & = ((1-\gamma^{2\tau})\Sigma^{-1}(\calW)+\bar{I})^2- \left((1-\gamma^{2}){N} +\frac{1}{s_N^2}I(v)\right)^2+4\gamma^{2\tau}\Sigma^{-1}(\calW)\bar{I} - \frac{4\gamma^2 N}{s_N^2}I(v).
\end{align*} Factorize $
     ((1-\gamma^{2\tau})\Sigma^{-1}(\calW)+\bar{I})^2- \left((1-\gamma^{2}){N} +\frac{1}{s_N^2}I(v)\right)^2$ as 
\begin{align*}
     \lp ((1-\gamma^{2\tau})\Sigma^{-1}(\calW)+\bar{I}) - \left((1-\gamma^{2}){N} +\frac{1}{s_N^2}I(v)\right)  \rp\lp (1-\gamma^{2\tau})\Sigma^{-1}(\calW)+\bar{I} + \left((1-\gamma^{2}){N} +\frac{1}{s_N^2}I(v)\right)  \rp 
\end{align*}
where the first factor is just $A$ from \eqref{eq:A} and thus $\Ocal(\tau/s_N^2)$. The second factor has
\begin{align*}
    \lp (1-\gamma^{2\tau})\Sigma^{-1}(\calW)+\bar{I} + (1-\gamma^{2}){N} +\frac{1}{s_N^2}I(v)  \rp & = \lp 2(1-\gamma^2)N+\bar I+\frac{1}{s_N^2}I(v) \rp \\
    & =\Ocal(1+\tau/s_N^2),
\end{align*} for all $\tau\gg 1.$ 
Thus
\begin{align*}
     B_1^2-B_2^2 -(4\gamma^{2\tau}\Sigma^{-1}(\calW)\bar{I} - \frac{4\gamma^2 N}{s_N^2}I(v))&= ((1-\gamma^{2\tau})\Sigma^{-1}(\calW)+\bar{I})^2- \left((1-\gamma^{2}){N} +\frac{1}{s_N^2}I(v)\right)^2 \\ & = \Ocal\lp\frac{\tau}{s_N^2}\lp(1+\frac{\tau}{s_N^2}\rp\rp.   \end{align*}  
for all such $\tau$ (in particular $\tau=s_N\gg 1$). 

Using the explicit expressions in Lemma  \ref{lem:algebra} we have that  
$$\left|4\gamma^{2\tau}\Sigma^{-1}(\calW)\bar{I} - \frac{4\gamma^2 N}{s_N^2}I(v)\right|=0=\Ocal\lp \frac{\tau}{s_N^2}\rp.$$
Now, we show that $B_1+B_2\geq 1$. To that end, observe that 
\begin{align*}
    B_1+B_2 & = \sqrt{ ((1-\gamma^{2\tau})\Sigma^{-1}(\calW)+\bar{I})^2 + 4 \gamma^{2\tau} \Sigma^{-1}(\calW)\bar{I}}+\sqrt{ \left((1-\gamma^{2}){N} +\frac{1}{s_N^2}I(v)\right)^2 + \frac{4 \gamma^{2}N}{s_N^2}I(v)}\\
    & \geq \sqrt{ \left((1-\gamma^{2}){N} +\frac{1}{s_N^2}I(v)\right)^2 + \frac{4 \gamma^{2}N}{s_N^2}I(v)}\\
    & =\Omega\lp 1+ \frac{\sqrt{N}}{s_N}\rp .
\end{align*}
Thus $$B = \Ocal\lp \frac{\tau}{s_N^2}\frac{ 1+ \tau/s_N^2}{1+\sqrt{N}/s_N}\rp$$
and therefore for $\tau\lesssim \sqrt{N}$ and recalling $s_N=\Omega(1)$
, 
\begin{align*}
    |\bar{J}-\JGF| \leq A+B = \Ocal\lp\frac{\tau}{s_N^2}\rp.\numberthis\label{eq:Jgap1}
\end{align*} 
Finally, recall from \cref{eq:tildeJJgap} that for $s_N\lesssim \sqrt{N}$ and with $\tau=s_N$ 
\begin{align*}
    \bar{J}=\Theta\left(\frac{\sqrt{N}}{s_N}\right)
\end{align*}
so that 
\begin{align*}J^{GF}=\Omega\lp\frac{\sqrt{N}}{s_N} - \frac{1}{s_N}\rp =\Omega\lp \frac{\sqrt{N}}{s_N}\rp. 
\end{align*}
Dividing \cref{eq:Jgap1} by $\bar{J}J^{GF}$, we have for $\tau\lesssim \sqrt{N}$ and $s_N=\Omega(1)$
\begin{align*}
|1/\bar{J} - 1/ J^{GF} | & = \frac{|\bar{J} - J^{GF}|}{\bar{J}J^{GF}} = \Ocal\lp \frac{\tau/s_N^2}{N/s_N^2}\rp =\Ocal\lp \frac{\tau}{N}\rp.\numberthis\label{eq:Jgapfinaform}
\end{align*}
Finally, setting $\tau=s_N$ and using $s_N\lesssim \sqrt{N}$, we have that $s_N/N = \Ocal(1/s_N)$ so that  
\begin{align*}
    |1/\bar{J} - 1/ J^{GF} | & =\Ocal\lp\frac{1}{s_N}\rp.
\end{align*}

\noindent \textbf{Step III:} $|1/J_\infty^{GF}-P_\infty|=0$ by definition. 

Combining all three bounds into \cref{eq:variancebound}, we get with $\tau=s_N$ that the MSE of our proposed filter in stationarity satisfies
\[
\hat{P}_t:=\Ex[(\hat{x}_t-X_t)^2] = \frac{1}{\bar{J}}\pm \Ocal\left(\frac{1}{s_N}\right)\numberthis
\] 
This completes the proof of Theorem \ref{thm:upperbound}. 

\vspace*{0.5cm}

\section{The centered Goggin filter upper bound\label{sec:GFcenteredUB}} 

As we mentioned previously in Section \ref{sec:GF}, Theorem \ref{thm:upperbound} is insufficient when $s_N\ll N^{1/4}$. We strengthen the upper bound with additional assumptions in the following Theorem.

\begin{theorem}\label{thm:upperboundcentered} Suppose $1\ll s_N\lesssim \sqrt{N}$. Let $\hat{x}_t$ be the (random) estimator produced by the centered Goggin filter in Definition \ref{def:centered-goggin-filter} and $\bar J_{\infty}$ be the solution to the equation 
\begin{align*}
\bar{J}_{\infty} = {\bf e}'I(\mathbb{V}){\bf e} + \frac{\bar{J}_{\infty}\Sigma^{-1}(\calW)}{\bar{J}_{\infty}+\gamma^{2\tau}\Sigma^{-1}(\calW) },
\end{align*}  with $\tau=s_N$.  Then, under the basic assumptions and under Assumption \ref{asum:dissipativity}, in stationarity, 
\begin{align} 
\Ex[\hat{x}_t^c-X_t] &= \Ocal\left(\frac{1}{\sqrt{N}}\right) \tag{Bias} \\
\Ex[(\hat{x}_t^c-X_t)^2]& = \frac{1}{\bar J_{\infty}}\pm \Ocal\left(\frac{\sqrt{s_N}}{N^{3/4}}\right).\tag{Variance} \end{align} 
\end{theorem}

\vspace*{0.5cm} 
\noindent \textbf{Proof of Theorem \ref{thm:main}:} Corollary \ref{thm:CRLB}, and Theorem \ref{thm:upperboundcentered} prove Theorem \ref{thm:maincentered}. \hfill \qed
\vspace*{0.5cm} 

The rest of this section is dedicated to the proof of Theorem \ref{thm:maincentered}.We start with a preliminary bound. With this lower bound the proof of the bias and variance bounds is similar to the one for the non-centered version of the filter.

\begin{lemma}\label{lem:preliminary}[preliminary bound] Suppose that $1\ll s_N\lesssim \sqrt{N}$ and that, in addition to our basic assumptions, also Assumption \ref{asum:centered} holds. 
Let $e_t := X_{t}-\hat x^{c}_{t}$, 
$P_t^c:=\Ex[e_t^2], \Upsilon_t^c = \Ex[|e_t|^3], F_t^c = \Ex[e_t^4]$.  Then \begin{align*}
P_t^c = \Ocal\lp\frac{s_N}{\sqrt{N}}\rp, \Upsilon_t^c = \Ocal \lp\frac{s_N^{3/2}}{N^{3/4}}\rp, F_t^c =\Ocal\lp\frac{s_N^2}{N}\rp. \numberthis \label{eq:lemma-bound}
\end{align*}
\end{lemma}
\vspace*{0.5cm} 

\textbf{Bias:} The proof proceeds initially similarly to that of Theorem \ref{thm:upperbound}. Let 
 $\tilde{x}_0=X_0=0$ and 
\begin{align*}
    \tilde{x}_t & = (1-K_t I(v))\gamma \tilde{x}_{t-1} +K_t(I(v)X_t + s_N\phi(v_t)).
\end{align*}
Then,  $\expec[\tilde x_t-X_t]=0$, and so it suffices to prove that $\expec[\hat x^{c}-\tilde x]=\Ocal(1/\sqrt{N})$. A second-order Taylor expansion gives 
$$\phi\lp \frac{Y_t-\gamma \hat{x}_{t-1}^c}{s_N}\rp= \phi(v_t)  +\frac{1}{s_N}\phi'(v_t) (Y_t-\gamma \hat{x}_{t-1}^c) + \frac{1}{s_N^2}\phi''(\varphi_t)(Y_t-\gamma \hat{x}_{t-1}^c)^2, 
$$ where $\varphi_t \in [\min\{v_t,v_t+(X_t-\gamma \hat{x}_{t-1}^c)/s_N\},\max\{v_t,v_t+(X_t-\gamma \hat{x}_{t-1}^c)/s_N\}].$
Thus,
$$\hat{x}_t^c = \gamma (1-K_tI(v))\hat{x}_{t-1}^c + K_t (I(v)X_t + s_N\phi(v_t)) + K_t\varrho_t,$$ where 
\begin{align*} \varrho_t&:= s_N\phi\lp\frac{Y_t-\gamma\hat{x}_{t-1}^c}{s_N}\rp - s_N\phi(v_t)-\Ex[\phi'(v_t)](X_t-\gamma\hat{x}_{t-1}^c)\\ & ~= 
(X_t-\gamma \hat{x}_{t-1}^c)(\phi'(v_t)- \Ex[\phi'(v_t)])+ \frac{1}{2} \phi''(\varphi_t)\frac{1}{s_N}(X_t-\gamma \hat{x}_{t-1}^c)^2.\end{align*}
Then, 
\begin{align*}
    \hat{x}^{c}_t - \tilde{x}_t 
    &= \gamma (1-K_tI(v))(\hat{x}_{t-1}^c-\tilde{x}_{t-1}) +K_t\varrho_t
\end{align*}

Defining $f(t):=|\expec [\hat x_{t}^c-\tilde x_t]|$, and noting that---because $v_t$ is independent of $X_t$ and $\hat{x}_{t-1}^c$--- $\Ex[(X_t-\gamma\hat{x}_{t-1}^c)(\phi'(v_t)-\Ex[\phi'(v_t)])]=0$, we have
\begin{align*}
    f(t) 
    & \leq \lp1-K_tI(v)\rp\gamma f(t-1)+K_t
\|\phi''\|_{\infty} \frac{1}{s_N}\Ex[(X_t-\gamma\hat{x}_{t-1}^c)^2].\numberthis\label{eq:hatminustildecentered}
\end{align*}

Writing $X_t = \gamma X_{t-1} +w_{t-1}/\sqrt{N}$, and recalling that $\Ex[w_1^2]=1$, we have that 
$$\Ex[(X_t-\gamma\hat{x}^c_{t-1})^2]\leq 2\gamma^2 \Ex[(X_{t-1} -\hat{x}_{t-1}^c)^2] + \frac{2}{N} =\Ocal\lp \frac{s_N}{\sqrt{N}}\rp,$$ where Lemma \ref{lem:preliminary} gave us $\Ex[(X_{t-1} -\hat{x}_{t-1}^c)^2]:=P_{t}^c = \Ocal(s_N/\sqrt{N})$ and we use $s_N=\Omega(1)$ so that $s_N/\sqrt{N}=\Omega(1/N)$. Lemma  \ref{lem:KFPXg} gives us $K_t=\Ocal(1/(s_N\sqrt{N}))$ so that
\begin{align*}
    f(t) \leq \lp1-K_tI(v)\rp\gamma f(t-1)+\Ocal\lp 
        \frac{1}{{Ns_N}}\rp.
\end{align*}
Iterating the recursion we get that, with $f(0)=\Ocal(1)$ 
\begin{equation} |\Ex[\hat{x}_t-X_t]| =: f(t) = \Ocal\lp \frac{1/(Ns_N)}{1/(s_N\sqrt{N})}\rp = \Ocal\lp\frac{1}{\sqrt{N}}\rp\label{eq:biascentered}\end{equation}  as . 

 \textbf{Variance:} We Our strategy is identical to the one in \cref{eq:variancebound}. We write 
\begin{align}
    |P_\infty^c-1/\bar{J}_{\infty}| \leq  | P^c_\infty-P_\infty|+|P_\infty-1/J_\infty^{GF}|+|1/J_\infty^{GF}-1/\bar J_\infty|\label{eq:centeredvariancebound}
\end{align} 
and bound $|\hat P^c_\infty-P_\infty|$. The second term is $0$ and the bound for the third term has 
$$|1/J_\infty^{GF}-1/\bar J_\infty| = \Ocal(s_N/N),$$ exactly as in \cref{eq:Jgapfinaform} with the choice $\tau=s_N$. Thus, for the variance bound it remains only to prove that $|P_{\infty}^c-P_{\infty}|=\Ocal(s_N/N).$

Recall $$\hat{x}_t^c = \gamma (1-K_tI(v))\hat{x}_{t-1}^c + K_t (I(v)X_t + s_N\phi(v_t)) + K_t\varrho_t,$$
so that 

\begin{align*} X_t - \hat{x}_t^c = (1-K_tI(v))\left(\gamma(X_{t-1}-\hat{x}_{t-1}^c) + \frac{1}{\sqrt{N}}w_{t-1}\right) - K_ts_N\phi(v_t) - K_t\varrho_t\end{align*} 

Exactly as in the proof of Theorem \ref{thm:upperbound}, using the notation 
$$R:=s_N^2I(v), ~~ Q = 1/N,$$ we have 
\begin{align}  P_t^c & = \Var(X_t-\hat{x}_t) \nonumber\\
& = (1-K_tI(v))^2 [\gamma^{2}P_{t-1}^c + Q] 
+K_t^2 R + K_t^2\Var(\remainder_t) + 2Cov(K_t\remainder_t,A_t), \label{eq:decompVarc}
\end{align}
where $$A_t:=  
(1-K_tI(v))\lb\gamma (X_{t-1}-\hx_{t-1}^c)+\frac{1}{\sqrt{N}}w_{t-1}\rb - K_ts_N\phi(v_t).$$

\begin{lemma} 
$$K_t^2\Var(\remainder_t) + 2|Cov(K_t\remainder_t,A_t)| = \Ocal\left(\frac{1}{\sqrt{s_N}N^{\frac{5}{4}}}\right).$$ \label{lem:remainderterms2}
\end{lemma}

Comparing \eqref{eq:decompVarc} with
\begin{align*}
    P_t = (1-K_tI(v))^2 (\gamma^{2} P_{t-1}+ Q) +K_t^2R
\end{align*}
we have 
\begin{align*}
   | P_t^c-P_t| & \leq (1-K_tI(v))^2\gamma^{2}|P_{t-1}^c - P_{t-1}|+\Ocal\lp\frac{1}{\sqrt{s_N}N^{5/4}}\rp.
\end{align*}
Iterating this recursion with $| P_0^c-P_0|=\Ocal(1)$ we have 
\begin{align*}
    | P_t^c-P_t| = \Ocal\lp \frac{(1/(\sqrt{s_N} N^{5/4})}{1/( s_N \sqrt{N})}\rp = \Ocal\lp\frac{\sqrt{s_N}}{N^{3/4}}\rp
\end{align*}
which completes our proof.
\section{Concluding Remarks} 
 
In this paper we revisited Goggin's filter for the non-Gaussian state-space model. We mapped the parameter space into filtering regimes, identifying and focusing attention on a {\em balanced regime}. In that regime, we quantified the sub-optimality rate---as an explicit function of the observation noise---of the Goggin filter which applies the Kalman filter to score-transformed observations. 

Our paper complements and renews attention to earlier literature on the non-Gaussian linear state-space model. Our results do not yet provide complete coverage. We do not cover the regime where the observation noise is small but not smaller than the driving noise, namely $\frac{1}{\sqrt{N}}\lesssim s_N \lesssim 1.$ In addition, our regularity assumptions are not uniform across regimes. To treat the range $s_N\lesssim N^{1/4}$ we introduced the centered variant of the GF and, for its analysis, imposed an additional dissipativity condition on the score function of the observation noise. 

Goggin's most general results allow for some nonlinear systems. For these, her result does not suggest a concrete filter. Instead, it shows that the optimal filter (which is the conditional expectation of the state given the expectations) converges to the conditional expectation in a suitably defined filtering problem for a nonlinear diffusion process. To the extent that tractable filters exist for that diffusion filtering problem, it is plausible that our analysis can extend to establish sub-optimality gaps for some nonlinear problems. 

Because our approach relies on non-asymptotic tools---a Cram\'er--Rao lower bound and CLT convergence rates for the Fisher information---it may be useful in expanding the coverage from estimation to the control of Partially Observed Markov Decision Processes (MDPs). A natural starting point is the quadratic regulator: with Gaussian dynamics and observations, the optimal policy satisfies a separation principle with the Kalman filter as the state estimator. Our analysis may help quantify how non-Gaussianity affects this separation and how it interacts with both the stochastic primitives and the control problem's cost parameters. 


\section*{Acknowledgements}

The authors are grateful to Vikram Krishnamurthy for his help during the development of this manuscript, and to Varun Gupta for motivating us to consider the centered filter. Imon Banerjee acknowledges IEMS Alumni Fellowship at Northwestern University for funding during the period at which this research was conducted.

\appendices

\section{Proofs of Lemmas \ref{lem:toonoisy}, \ref{lem:nonoise}, \ref{lem:balanced}, \ref{lem:KF}}\label{sec:prf-regimelemmas}
  \begin{proofof}{Lemma \ref{lem:toonoisy}} Specializing \cite[Equation (2)]{tichavsky1998posterior} we get that, in stationarity, $MSE^{\star}\geq 1/J$ where  
  \begin{align}
      J:=\frac{1}{2}\left(\cR+(1-\gamma^2)\cQ + \sqrt{(\cR+(1-\gamma^2)\cQ)^2+4\gamma^2 \cQ\cR}\right),\label{eq:J-Def}
  \end{align}  where $\cQ := I(w)/(1/\sqrt{N})^2=N I(w)$ and $\cR:=I( v)/s_N^2$. 

Under the hypothesis of the lemma $s_N\gg \sqrt{N}$ so that $\cR= \ocal(1)$ and $\cQ\cR=\ocal(1)$. In turn, 
$$
\sqrt{(\cR+(1-\gamma^2)\cQ)^2+4\gamma^2 \cQ\cR}\leq \cR+(1-\gamma^2)\cQ+ 2\gamma \sqrt{\cQ\cR} = (1-\gamma^2)\cQ\pm \ocal(1).
$$
Thus, 
$$
J\leq (1-\gamma^2)\cQ + \ocal(1),
$$
so that 
\[MSE^{\star}\geq \frac{1}{J}\geq \frac{1/\cQ}{1-\gamma^2} -\ocal\left(1\right)=\frac{1}{2I(w)}\pm \ocal(1).\numberthis\label{eq:MSE1form}
\]
Because $1/I(w)\leq Var(w)=1$ with equality only in the Gaussian case 
$$\frac{1}{2I(w)} < \frac{1}{2}.$$

This shows that the direct CRLB lower bound is too loose relative to the lemma's statement. To establish the lemma, we resort to Theorem \ref{thm:CRLB}. Setting $\tau=N$ and $s_N\gg \sqrt{N}$ there we have that $$\bar{J}_{\infty} = \Theta\lp \max\left\{\frac{\tau}{s_N^2}+\frac{\sqrt{N}}{s_N},1\right\}\rp =\Theta(1)$$ and
$$MSE^{\star} \geq \frac{1}{\bar{J}_{\infty}} -\Ocal\lp\frac{1}{N\bar{J}_{\infty}}\rp=\frac{1}{\bar{J}_{\infty}} \pm  \Ocal\lp\frac{1}{N}\rp.$$ 
We now analyze $\bar J_\infty$ for its asymptotic value under the assumption that $s_N\gg \sqrt{N}.$
Recall from \cref{eq:rec-barJinf} that 
\begin{align*}
    \bar J_{\infty} & = \frac{{\bf e}'I(\mathbb{V}){\bf e}+\Sigma^{-1}(\calW)(1-\gamma^{2\tau})+\sqrt{\bigl({\bf e}'I(\mathbb{V}){\bf e}+\Sigma^{-1}(\calW)(1-\gamma^{2\tau})\bigr)^2+4{\bf e}'I(\mathbb{V}){\bf e}\gamma^{2\tau}\Sigma^{-1}(\calW)}}{2}.
\end{align*}

Recall that $\Sigma(\calW)=Var(\calW) = 
 \frac{1-\gamma^{2\tau}}{N(1-\gamma^2)}.$ Setting $\tau=N$, 
\begin{align*}
\Sigma^{-1}(\calW)(1-\gamma^{2\tau}) = N(1-\gamma^2) = 2\pm\ocal(1).
\end{align*}

Recalling that, in the setting of this lemma,
$s_N\gg \sqrt{N}$ we have with $\tau=N$ that $${\bf e}'I(\mathbb{V}){\bf e}=\Theta\lp \frac{\tau}{s_N^2}\rp= \ocal(1).$$ 
Thus, 
$$\bar J_{\infty}=\frac{\ocal(1)+2 + \sqrt{(2+\ocal(1))^2 + \ocal(1)} }{2}=2+\ocal(1),$$ so that
$$\frac{1}{\bar{J}_{\infty}} = \frac{1}{2}\pm\ocal(1),$$ and we conclude $$MSE^{\star} \geq  \frac{1}{2}-\ocal(1),$$ as required. \end{proofof}

\begin{proofof}{Lemma \ref{lem:nonoise}} 

First, let us prove that 
the Cram\'er--Rao lower bound is too loose in the setting of this lemma. As before  $\cQ := I(w)/(1/\sqrt{N})^2=N I(w)$ and $\cR:=I(v)/s_N^2$.  Dividing $J$ by $\cR$ we get 
  $$\frac{J}{\cR} = \frac{1}{2}\lp 1 + (1-\gamma^2)\frac{\cQ}{\cR} + \sqrt{\lp1+(1-\gamma^2)\frac{\cQ}{\cR}\rp^2+4\gamma^2\frac{\cQ}{\cR}}\rp.$$ Under the assumptions of this lemma, specifically that $s_N=\ocal( 1/\sqrt{N})$, $\cQ=\ocal(\cR)$, so that $\cQ/\cR = o(1)$ and  $J/\cR = 1+\ocal(1)$. In turn, $J=\cR+\ocal(\cR)= \frac{I(v)}{s_N^2}+\ocal\lp \frac{I(v)}{s_N^2}\rp$. Therefore, the CRLB gives
  \begin{align}\label{eq:CRLBnonoise} 
     \frac{1}{s_N^2} \times  MSE^{\star}\geq  \frac{1}{I(v)}- \ocal\left(1\right).
  \end{align}
  Since $1/I(v) \leq 1/\Var(v) =  1$ with equality only with Gaussian $v$, we have generally $\frac{1}{I(v)}<1$ so that the CRLB-based lower bound in \eqref{eq:CRLBnonoise} is 
too loose relative to the lemma's statement. 

Instead of the CRLB we will use the following lemma to prove the MSE lower bound. Its proof is deferred to Section \ref{sec:auxproofs}. 
    
\begin{lemma}\label{lem:lecamapp} Suppose we have an observation of a parameter $\theta\in\mathbb{R}$:
 $$Z = \theta + \zeta, $$ where $\zeta=\sigma\zeta_0$ for fixed $\sigma>0$ and a zero mean, unit variance and continuous random variable $\zeta_0$ with support on the real line and finite fourth moment: $\Ex[\zeta_0]=0,Var(\zeta_0)=1$ and $\Ex[\zeta_0^4]<\infty$. Assume a fixed Lipschitz prior $\pi(\theta)$ that has $\pi(0)>0$. Then, for any estimator $\hat{\theta}$ of $\theta$, as $\sigma^2\downarrow 0$, 
$$\Ex[(\hat{\theta}-\theta)^2]\geq \sigma^2 -\ocal(\sigma^2), $$ where, the expectation is over the joint distribution of the prior and the observation.
\end{lemma} 
\vspace*{0.2cm} 
  Equipped with this lemma, let us recall our linear system 
\begin{align*} X_{t+1}&=\gamma X_t+\frac{1}{\sqrt{N}}{w}_t,\numberthis\label{eq:lemnonoise-eq1} \\
Y_{t} & = X_t + s_N {v}_t\numberthis\label{eq:lemnonoise-eq2}\end{align*} 
{
Given a fixed (arbitrary) filter, fix a time $t-1$, and let $\hat{x}_{t-1}$ be the estimate of $X_{t-1}$ at time $t-1$. 
}
Then, the prior at time $t$ (before the observation at time $t$) is given by the convolution of the independent variables
$\gamma \hat{x}_{t-1}$ and ${w}_{t-1}$. That is
\begin{align*}
    \hat{x}_{t|t-1} = \gamma  \hat{x}_{t-1} +\frac{1}{\sqrt{N}}{w}_{t-1}.\numberthis\label{eq:lemnonoise-eq3}
\end{align*}

Under the assumptions of this lemma, ${w}_t$ is Lipschitz continuous; let $L_{w}$ be the corresponding Lipschitz constant. The convolution in \cref{eq:lemnonoise-eq3} inherits the Lipschitz constant and bound on the density from ${w}_t$. Indeed, take a random variable $S_1$ and a continuous random variable $S_2$ that has a strictly positive, bounded and Lipschitz continuous density  $f_{2}$. Then, 
$$f_{1+2}(s) = \int_{s_1}f_2(s-s_1)dF_1(s_1)\leq |f_2|_{\infty}\int_{s_1}dF_1({s_1})=|f_2|_{\infty}.$$ 
Similarly, 
\begin{align*} |f_{1+2}(s)-f_{1+2}(s')|& =\left| \int_{s_1}f_2(s-s_1)dF_1(s_1)-\int_{s_1}f_2(s'-s_1)dF_1(s_1)\right|
\\ & \leq   \int_{s_1}|f_2(s-s_1)-f_2(s'-s_1)|dF(s_1)\leq L_2|s-s'|,
\end{align*}  where $L_2$ is the Lipschitz constant of $f_2$. 

In order to apply Lemma \ref{lem:lecamapp}, we next verify that the density of $\hat x_{t|t-1}$ is positive at $0$. 
Following the line of arguments before, suppose $f_2(0)>0$. Then we get 
\begin{align*}
    f_{1+2}(0) = \int_{s_1}f_2(-s_1)dF_1(s_1)=\Ex[f_2(-S_1)]>0,
\end{align*}  where the last inequality follows since $\mathbb{P}\{f_2(-S_1)>0\}=1$ by the assumption of $f_2$. 

We conclude that the distribution $\hat{X}_{t|t-1}$ inherits the boundedness, positivity at $0$, and Lipschitz continuity of the density from ${w}_t$. Letting $Z=Y_t$, $\sigma=s_N$ and $\pi$ be the density of $\hat{X}_{t|t-1}$, Lemma \ref{lem:lecamapp} give us that no estimator can achieve MSE smaller than $s_N^2-\ocal(s_N)$. \end{proofof}

\begin{proofof}{Lemma \ref{lem:balanced}} For the KF, the steady-state information $J=1/P$ can be derived by solving the update equations for $1/P$. It admits a closed form solution given by
 \begin{align*}
      J=\frac{1}{2}\left(\cR+(1-\gamma^2)\cQ + \sqrt{(\cR+(1-\gamma^2)\cQ)^2+4\gamma^2 \cQ\cR}\right)
  \end{align*}
with $\gamma=(1-1/N)$, $\cQ=N$, and $\cR=1/s_N^2$. In turn, 
\begin{align*}
    J\geq \frac{1}{2}\lp (1-\gamma^2)\cQ + \sqrt{((1-\gamma^2)\cQ)^2 + 4\gamma^2 \cR\cQ}\rp
\end{align*}
which we rearrange as
\begin{align*}
    \frac{J}{ (1-\gamma^2)\cQ }\geq \frac{1}{2}\lp1 + \sqrt{1 + 4\gamma^2 \frac{\cR\cQ}{ (1-\gamma^2)^2\cQ^2 }}\rp.
\end{align*}

Substituting the value of $\cQ,\cR$ we have
\begin{align*}
\frac{\cR}{(1-\gamma^2)^2\cQ} 
&= \frac{1}{N s_N^2\left(\frac{2}{N}-\frac{1}{N^2}\right)^2} = \frac{1}{\frac{s_N^2}{N}\left(2-\frac{1}{N}\right)^2}=\Omega(1),
\end{align*}
with which we conclude 
$$\frac{J}{(1-\gamma^2)\cQ} \geq \frac{1}{2}\left(1+ \sqrt{1+\Omega(1)}\right) = 1+\Omega(1).$$
Therefore, there exists $\delta>0$ such that $J\geq (1+\delta)\cQ(1-\gamma^2)$ so that for $\epsilon=\delta/(1+\delta)>0$ 
$$MSE^{KF} = P = \frac{1}{J}\leq (1-\epsilon)\frac{1}{\cQ(1-\gamma^2)}.$$
Recall from \cref{eq:MSE1form} that $MSE\pow 1\geq  1/(\cQ(1-\gamma^2))-\ocal(1)$. Since $\epsilon>0$, this shows that 
\begin{align*}
    MSE\pow 1-MSE^{KF}>\epsilon MSE\pow 1=\Omega\lp MSE\pow 1\rp.\end{align*} 
The argument is almost identical for the other bound noting that 
$$\frac{J}{\cR}\geq \frac{1}{2}\lp 1 + \sqrt{1 + 4\gamma^2 \frac{\cQ}{\cR}}\rp. $$
Using finally the lemma's assumption that $s_N\lesssim \sqrt{N}$, so that $\cQ/\cR = N/s_N^2 = \Omega(1)$, we conclude as before that $J\geq (1+\delta)R$ so that $P\leq (1-\epsilon)s_N^2$. In turn, because $MSE\pow 2 \geq s_N^2-\ocal(s_N^2)$,  
$$MSE\pow 2 -MSE^{KF} =\Omega(MSE\pow 2).$$
\end{proofof}

\begin{proofof}{Lemma \ref{lemma:log-concave}} We are considering densities of the form $f(x) = \frac{1}{Z}e^{-V(x)}$. The score is then $$\phi(x)  = -\frac{f'(x)}{f(x)} = V'(x).$$ By assumption $|V'(x)| \leq A+B|x|$ which gives immediately the growth requirement in Assumption \ref{asum:dissipativity}. Next, because $\phi'(x) = V''(x)$, and our assumption that $|V'''(x)|\leq C$, $|V''(x)| \leq  |V''(0)| + C|x|$. Then, $|\frac{\partial}{\partial y}\phi(v_t+y)|=|\phi'(v_t+y)|\leq C+D|v_t|$, for some $C,D>0$. Since $\Ex[|v_t|]<\infty$ we can interchange differentiation and expectation and write 
\[
h'(y) \;=\; \expec\bigl[\phi'(v_t + y)\bigr].
\]
By assumption, $\inf_{z}\phi'(z)=\inf_{z} V''(z)\geq m>0$ for some $m>0$, hence
\begin{equation}\label{eq:hprime_lower_bound}
h'(y) \;=\; \expec[\phi'(v_t + y)] \;\ge\; m \qquad,  \text{ for all } y \in \mathbb{R}.
\end{equation}
Since $\phi$ is the score function,
\[
h(0) = \expec[\phi(v_t + 0)] = \expec[\phi(v_t)] = 0.
\]

\medskip
\noindent\textbf{Strong dissipativity:} Fix $y > 0$. By the fundamental theorem of calculus and \eqref{eq:hprime_lower_bound},
\[
h(y) - h(0) = \int_0^y h'(s)\,\mathrm{d}s \;\ge\; \int_0^y m\,\mathrm{d}s \;=\; m y.
\]
Multiplying both sides by $y > 0$ gives
\[
h(y)\,y \;\ge\; m\,y^2, \qquad y > 0.
\]

Fixing $y < 0$ we have, similarly,  
\[
h(0) - h(y) = \int_y^0 h'(s)\,\mathrm{d}s \;\ge\; \int_y^0 m\,\mathrm{d}s \;=\; -m y.
\]
so that $h(y) \le m y$ for all $y < 0$. Multiplying by $y < 0$ reverses the inequality:
\[
h(y)\,y \;\ge\; m\,y^2, \qquad y < 0.
\]

For $y = 0$, the inequality holds with equality. Combining all three cases, we conclude that 
\[
h(y)\,y \;\ge\; m\,y^2 \qquad \text{for all } y \in \mathbb{R},
\]
as required with $\zeta=m$ in Assumption \ref{asum:dissipativity}. 

Finally, note that $\phi''(x) = V'''(x)$ and hence, $\|\phi''\|_{\infty}\leq \|V'''\|_{\infty} \leq C$. Also, since $V''(x) \geq m>0,$ there exists $a,b>0$ such that $V(x) \geq b x^2$ for all $x:|x|\geq a$. In turn, for any moment $k$ (in particular $k=4$) $$\Ex[v_1^4]\leq \Gamma \left(a^4 + \int_{x:|x|\geq a}x^4 e^{-V(x)}\right)<\infty.$$ This guarantees that Assumption \ref{asum:primitives-observation} is, as well, satisfied. \end{proofof}

\noindent \begin{proofof}{Lemma \ref{lem:KF}} Let $P^{KF}$ be the steady-state $MSE$ of the KF. Then, $J^{KF}$ satisfies equation 

$$J^{KF} = \frac{1}{s_N^2} + \frac{J^{KF} N}{J^{KF}+\gamma^2 N}.$$

For $1\ll s_N\lesssim \sqrt{N}$, as assumed,  $s_N^2/N=\Ocal(1)$ so that taking $\tau=\sqrt{N}$ in Theorem \ref{thm:CRLB}, we have  $MSE^{\star}\geq \frac{1}{\bar{J}}-\Ocal\lp \frac{1}{\sqrt{N}}\rp$ where $\bar{J}$ (defined with $\tau=\sqrt{N}$) is the stationary point of the recursion in \eqref{eq:barJrecursion}. As in step II in the proof of Theorem \ref{thm:upperbound}--- see specifically \eqref{eq:Jgapfinaform} there 
we have that $|\bar{J} - J^{GF}|=\Ocal(\tau/N)=\Ocal(1/\sqrt{N})$ for $\tau=\sqrt{N}$.

Substituting $\tau=\Ocal(\sqrt{N})$ and $s_N\gg 1$  in \cref{eq:Jgapfinaform}
$|1/J^{GF}- 1/\bar{J}| =\Ocal(1/\sqrt{N}).$ If we can show that $$1/J^{KF}- 1/J^{GF} = \Omega(s_N/\sqrt{N}), 
\text{\, then \,} 1/J^{KF}- 1/\bar{J} = \Omega(s_N/\sqrt{N})-\Ocal(1/\sqrt{N})=\Omega(s_N/\sqrt{N})$$ for $s_N\gg 1$ and $$MSE^{KF} - MSE^{\star} \geq  1/J^{KF} - 1/\bar{J} +\Ocal(1/\sqrt{N}) = \Omega(s_N/\sqrt{N})=\Omega(MSE^{\star})$$ as stated.

Among random variables with the same variance, the Fisher information is minimized by the Gaussian with the given variance \cite{johnson2004information}. In other words, for a random variable $ v$ with unit variance ($\sigma_v^2= 1$), $I({v}) \geq 1$ and the equality is achieved only with the Gaussian distribution. In turn, if ${v}$ is non-gaussian we can define a strictly positive number (not dependent on N) $$\eta = I({v})-1>0.$$

The respective information quantities $J^{KF}$ and $J^{GF}$ solve the quadratic equations given earlier, whose solutions are
\begin{align*}
J^{KF} 
&= \frac{1}{2}\left((1-\gamma^2)N + \frac{1}{s_N^2}
+ \sqrt{\left((1-\gamma^2)N+\frac{1}{s_N^2}\right)^2
+ 4\gamma^2N\frac{1}{s_N^2}}\right), \\
J^{GF}
&= \frac{1}{2}\left((1-\gamma^2)N + \frac{I(v)}{s_N^2}
+ \sqrt{\left((1-\gamma^2)N+\frac{I(v)}{s_N^2}\right)^2
+ 4\gamma^2N\frac{I(v)}{s_N^2}}\right),
\end{align*}
where $1 < I(v)$.

Recalling that $1-\gamma^2 = 2/N + \mathcal{O}(1/N)$, we have
\[
(1-\gamma^2)N = \Theta(1), \qquad
\frac{1}{s_N^2} = o(1), \qquad
\frac{I(v)}{s_N^2} = o(1).
\]
Moreover,
\[
\gamma^2 N \frac{1}{s_N^2} = \Omega\!\left(\frac{N}{s_N^2}\right)=\Omega(1),
\qquad
\gamma^2 N \frac{I(v)}{s_N^2}=\Omega(1),
\]
for all $s_N\le \sqrt{N}$. Thus,
\begin{align}
\nonumber
J^{KF}
&= \frac{1}{2}\left((1-\gamma^2)N + o(1)
+ \sqrt{\left((1-\gamma^2)N+o(1)\right)^2
+ 4\gamma^2N\frac{1}{s_N^2}}\right), \\
\label{eq:comparisonJ}
J^{GF}
&= \frac{1}{2}\left((1-\gamma^2)N + o(1)
+ \sqrt{\left((1-\gamma^2)N+o(1)\right)^2
+ 4\gamma^2N\frac{I(v)}{s_N^2}}\right).
\end{align}

Taking a Taylor expansion of the square-root terms around
$4\gamma^2N/s_N^2$ and $4\gamma^2N I(v)/s_N^2$ respectively yields
\[
J^{GF}-J^{KF}
= \sqrt{4\gamma^2N\frac{I(v)}{s_N^2}}
- \sqrt{4\gamma^2N\frac{1}{s_N^2}}
+ o(1)
= \Omega\!\left(\bigl(\sqrt{I(v)}-1\bigr)\frac{\sqrt{N}}{s_N}\right)
= \Omega\!\left(\frac{\sqrt{N}}{s_N}\right),
\]
where the last equality follows from $I(v)-1>0$.

Using $(1-\gamma^2)N=\Theta(1)$, $N/s_N^2=\Omega(1)$, together with
\[
J^{KF}=\Theta\!\left(\frac{\sqrt{N}}{s_N}\right),
\qquad
J^{GF}=\Theta\!\left(\frac{\sqrt{N}}{s_N}\right),
\]
we obtain
\[
\left|\frac{1}{J^{KF}}-\frac{1}{J^{GF}}\right|
= \frac{|J^{KF}-J^{GF}|}{J^{KF}J^{GF}}
= \Omega\!\left(\frac{\sqrt{N}/s_N}{N/s_N^2}\right)
= \Omega\!\left(\frac{s_N}{\sqrt{N}}\right),
\]
as required.
\end{proofof}

\begin{proofof}{Lemma \ref{lem:algebra}} 

Using second order Taylor expansion of $x^{2\tau}$ around $x=1$ we have 

\begin{equation} \label{eq:ratiosimp}
    \frac{1-\gamma^{2\tau}}{1-\gamma^{2}} = \frac{\frac{2\tau}{N}+\Ocal(\tau^2 / N^2)}{\frac{2}{N}-\frac{1}{N^2}} = \tau +\Ocal\lp\frac{\tau^2}{N}\rp =\Theta(\tau).  
    \end{equation}

Since the information of a vector of independent random variables is the sum of the individual-coordinate Fisher information, we have that 
\begin{align*} 
 {\bf e}' I(\mathbb{V}){\bf e} & = \frac{1}{s_N^2}I(v)\sum_{s=1}^{\tau }\gamma^{-2(s-1)}\\ 
    & = \frac{1}{s_N^2}I(v) \frac{1-\gamma^{-2\tau}}{1-\gamma^{-2}}\\
    & = \frac{\gamma^2}{s_N^2\gamma^{2\tau}}I(v) \frac{1-\gamma^{2\tau}}{1-\gamma^{2}}\end{align*}

For all $\tau\leq N$ (in particular for $\tau=1$), $\gamma^{2\tau}=\Theta(1)$ so that also $\gamma^2/\gamma^{2\tau} =\Theta(1)$. In turn, for all $\tau\leq N$,
    \begin{align} {\bf e}' I(\mathbb{V}){\bf e}=\frac{\gamma^2}{s_N^2\gamma^{2\tau}}I(v) \frac{1-\gamma^{2\tau}}{1-\gamma^{2}} & =\Theta\lp\frac{\tau}{s_N^2}I(v)\rp+\Ocal\left(\frac{1}{s_N^2}\frac{\tau^2}{N}\right)\nonumber\\
    & =\Theta\left(\frac{\tau}{s_N^2}\right)\numberthis\label{eq:VBB-infobound}.
\end{align}

That $\Sigma^{-1}(\calW)$ is as stated follow from \eqref{eq:def-Wcal}. And that $\Sigma^{-1}(\calW)=\Theta(N/\tau)$ then follows from \eqref{eq:ratiosimp}. 

Finally, for any integer $k$ there exists $\Gamma_k>0$ such that $$\Ex[X_t^{2k}]=\sum_{i=0}^{t-1} \gamma^{3(t-1-i)}\Ex[w_i^{2k}/N^{k}]  \leq \Gamma_k \Ex[w_1^{2k}] \frac{1}{N^{k}(1-\gamma^3)} = \Theta\lp \frac{1}{N^{k-1}}\rp,$$ where we use the finite moments of all orders for $w$; see Assumption \ref{asum:primitives-driving}. \end{proofof}

\section{Proofs of Auxillary Lemmas\label{sec:auxproofs}}

\begin{proofof}{Lemma \ref{lem:CLT_gn}}

We recall that 
\[
\Wcal_k = \frac{1}{\sqrt{N}}\sum_{s\in [\tau]}\gamma^{s-1} w_{\tau k-s},
\] so that 
\begin{align*}
   \Var(\Wcal_k)  = \frac{1}{N}\sum_{s\in [\tau]}\gamma^{2(s-1)} 
   = \frac{1}{N}\frac{1-\gamma^{2\tau}}{1-\gamma^2}
   = \frac{\tau}{N}\pm\mathcal{O}\lp \frac{\tau^2}{N^{2}}\rp. \numberthis\label{eq:cltgnprf-eq2}
\end{align*}
The standardized Fisher information for a random variable $U$ is given by 
$$J_{st}(U) = \Var(U)I(U)-1.$$
For the following lemma $R^{\star}(U)$ is the Poincar\'e constant of a random variable $U$; see \eqref{eq:poincaredefin}. 
.  

\begin{lemma}[Proposition 3.2 in \cite{johnson2004information}]\label{prop:johnson_inforbd}
    Let $U_1\dots,U_n$ be independent, $0$ mean random variables with variances $\sigma_1^2,\dots,\sigma_n^2$ respectively. Then, the standardised Fisher information $J_{st}$ satisfies
    \[
    J_{st}\lp\frac{U_1+\dots+U_n}{\sqrt{n}}\rp\leq \sum_{i=1}^n \frac{\alpha_i}{1+c_i} J_{st}(U_i)
    \]
    where 
    \[
    \alpha_i=\sigma_i^2/(\sum_{j}\sigma_j^2) \text{ and } c_i = \frac{1}{2R^\star (U_i)}\sum_{j\neq i}\frac{1}{I(U_j)}.
    \]
\end{lemma} \vspace*{0.2cm} 
We substitute  $n=\tau$, $U_i = \gamma^{i-1}w_{\tau k -i}$ in Lemma \ref{prop:johnson_inforbd}. 
Recall that $\gamma=1-1/N$ and notice that $\gamma^{2\tau} = 1-2\tau/N + \Ocal(\tau^2/N^2)$ and similarly (take the special case $\tau=1$) $\gamma^2 = 1-2/N + \Ocal(1/N^2)$. Thus, if $\tau = o(N)$, there exists a constant $\Gamma$ for which $(1-\gamma^2)/(1-\gamma^{2\tau}) \leq \Gamma/\tau$. If $\tau=\Theta(N)$, then  $1-\gamma^{2\tau} \geq 1-\frac{1}{\Gamma}$ for a re-defined $\Gamma$. Overall, we have a constant $\Gamma$ such that 

\begin{align*}
    \alpha_i & = \gamma^{2(i-1)}/\Bigl(\sum_{i\in[\tau]}\gamma^{2(i-1)}\Bigr)= \frac{\gamma^{2(i-1)}(1-\gamma^2)}{1-\gamma^{2\tau}}\leq \frac{\Gamma}{\tau}, \\
    c_i & = \frac{1}{2R^\star (\gamma^{i-1}w_{\tau k -i})}\sum_{j\neq i}\frac{1}{I(\gamma^{j-1}w_{\tau k -j})}\\ 
    &= \frac{1}{2\gamma^{2(i-1)}R^\star (w_{\tau k -i})}\sum_{j\neq i}\frac{\gamma^{2(j-1)}}{I(w_{\tau k -j})}\\
    & =  \frac{1}{2\gamma^{2i}R^\star (w_1)}\sum_{j\neq i}\frac{\gamma^{2j}}{I(w_1)}, 
\end{align*}
where the second-to-last equality uses since $R^\star(\gamma^t w)=\gamma^{2t}R^\star(w)$ by \cite[Theorem 2, item v]{borovkov1984inequality} as well as that for a random variable $U$ and $a>0$, $I(a U) = I(U)/a^2$. 

Now we have, 
\begin{align*}
     \frac{1}{2\gamma^{2(i-1)}R^\star (w_1)}\sum_{j\neq i}\frac{\gamma^{2(j-1)}}{I(w_1)} = \frac{1}{2R^\star (w_1)I(w_1)}\lp {\frac{1-\gamma^{2\tau}}{\gamma^{2(i-1)}(1-\gamma^2)}}-1 \rp.\numberthis\label{eq:prfcltgn-eq1}
\end{align*}

Arguing as before, and using $\tau \gg 1$ we have a constant $\Gamma>0$ such that  $$\frac{1-\gamma^{2\tau}}{\gamma^{2(i-1)}(1-\gamma^2)}-1 \geq \frac{1}{\Gamma}\tau.$$
Overall, we have that 
\begin{align*}
    \frac{\alpha_i}{1+c_i} & \leq \frac{\Gamma}{\tau} \cdot \frac{1}{1+\frac{1}{2\Gamma R^\star (w_1)I(w_1)}\tau }\leq  \frac{2\Gamma^2 R^\star (w_1)I(w_1)}{2\Gamma \tau R^\star (w_1)I(w_1) + \tau^2} \leq \frac{2\Gamma^2 R^\star (w_1)I(w_1)}{\tau^2}.
\end{align*}

For a random variable $U$, $J_{st}(a U ) =J_{st}(U)$. In particular, we have both 
$J_{st}(\Wcal_k) = J_{st}(\Wcal_k/\sqrt{\tau})$ and $J_{st}(\gamma^{i-1}w_{\tau k -i})=J_{st}(w_{\tau k -i})=I(w_1)-1$. Therefore,
\begin{align*}
    J_{st}(\Wcal_k) = J_{st}(\Wcal_k/\sqrt{\tau}) \leq \sum_{i\in [\tau]}\frac{\alpha_i}{1+c_i}J_{st}(\gamma^i w_{\tau k -i-1})
    \leq \tau J_{st}(w_1)\frac{8R^\star (w_1)I(w_1)}{\tau^2}
    =\Ocal\lp \frac{1}{\tau}\rp,
\end{align*}
and we get $0\leq \Var(\Wcal)I(\Wcal) -1 =\Ocal(1/\tau)$ as stated. 
\end{proofof}

\begin{proofof}{Lemma \ref{lem:KFPXg}}

Recall from Definition \ref{def:goggin-filter} that with $Q=\frac{1}{N}$ and $R=s_N^2I(v)$
\begin{align*}
    P_t & = \frac{(\gamma^2P_{t-1}+Q)R}{R+I^2(v)(\gamma^2P_{t-1}+Q)}.
\end{align*}
In stationarity this admits the steady-state solution 
\begin{align*}
    P & = \frac{(\gamma^2P+Q)R}{R+I^2(v)(\gamma^2P+Q)}.
\end{align*}

This quadratic equation has the solution 

$$P = \frac{1}{2\gamma^2I^2(v)}\lp (1-\gamma^2)R + I^2(v)Q + \sqrt{((1-\gamma^2)R + I^2(v)Q)^2 + 4QR\gamma^2 I^2(v)}\rp  $$ 

Recalling that $\gamma=1-1/N$ so that $(1-\gamma^2) = \Ocal(1/N)$, as well as that $s_N \lesssim\sqrt{N} = o(N)$ we have that (i) $(1-\gamma^2)R + I^2(v)Q=\Theta(s_N^2/N)$, (ii) $((1-\gamma^2)R + I^2(v)Q)^2 =\Ocal(s_N^2/N)$, and (iii) $4QR \gamma^2 I(v) = \Theta(s_N^2/N).$  

Overall, $$P = \Theta(s_N^2/N + s_N/\sqrt{N}) = \Theta(s_N/\sqrt{N}),$$ where the last equality follows since $s_N \lesssim \sqrt{N}$. 

This in turn implies that, in stationarity,
\begin{align*}
    K_t = \frac{P_tI(v)}{R} =\Theta\left(\frac{1}{s_N\sqrt{N}}\right).
\end{align*}

For the upper bound on $\Ex[X_t^2]$ see Lemma \ref{lem:algebra}. This completes the proof.
\end{proofof}

\begin{proofof}{Lemma \ref{lem:remainderterms1}}
 
    Recall from \cref{eq:R_t} that
\begin{align*} 
\remainder_t  & =\underbrace{X_t(\phi'(v_t)-\Ex[\phi'(v_t)])}_{\text{Term 1}} 
\pm\underbrace{\|\phi''\|_{\infty} \frac{1}{s_N}(X_t)^2}_{\text{Term 2}}, \\
A_t & = (1-K_tI(v))[\gamma(X_{t-1}-\hx_{t-1})+ w_t/\sqrt{N}] - K_t s_N v_t.
\end{align*} 

We first show that the variance of $\remainder_t$ is $\Ocal(1)$. 
\begin{align*}
    \Var(\remainder_t) = \Var(\text{Term 1})+\Var(\text{Term 2})+2\Cov(\text{Term 1},\text{Term 2}).
\end{align*}

Because $X_t$ is independent of $v_t$, 
$\Ex[\text{Term 1}] = \Ex[X_t] \Ex[\phi'(v_t)-\Ex[\phi'(v_t)]]=0$ so that $\Cov(\text{Term 1},\text{Term 2}) = \Ex[\text{Term 1}\times \text{Term 2})]$. By 
H\"{o}lder’s inequality and using again the independence of $X_t$ from $v_t$, we have 
\[
|\Cov(\text{Term 1},\text{Term 2})| = |\Ex[\text{Term 1},\text{Term 2}]|\leq \frac{1}{s_N} 
\|\phi''\|_{\infty} \Ex[|X_t^3|]\Ex[|\phi'(v_t)-\Ex[\phi'(v_t)]|] = \Ocal(1/s_N). 
\] 
Here we used two facts. First, with $X_0=0$, $\Ex[X_t^{2k}] = \Ocal(1)$ for all integer $k$; see Lemma \ref{lem:algebra}. Second,  $|\phi'(v_t)-\Ex[\phi'(v_t)]|\leq \|\phi''\|_{\infty}\Ex[|v_t|]<\infty$. 

To bound the variance terms use again the independence of $v_t$ of $X_t$ and  $\expec[\text{Term 1}]=0$, to have 
\begin{align*}
  \Var(X_t(\phi'(v_t)-\Ex[\phi'(v_t)])) 
  &= \Ex[X_t^2(\phi'(v_t)-\Ex[\phi'(v_t)])^2]\\
  &= \Ex[X_t^2]\Var(\phi'(v_t))\\
   &\leq  \Ex[X_t^2]|\phi''|_{\infty}\Ex[|v_t|^2] =\Ocal(1).
\end{align*}
Term 2 is $\Ocal(1/s_N)$ by our bound for the moments of $X_t$.

We turn to $\Cov(K_t\remainder_t,A_t)=K_t\Cov(\remainder_t,A_t)$. 
\[
\Cov(\remainder_t,A_t) = \expec[\remainder_tA_t]-\expec[\remainder_t]\expec[A_t].
\]
And we bound each of the previous terms separately, beginning with $\expec[\remainder_t]$.

By an argument similar to the one in the previous part,
\begin{align*}
    \expec[|\remainder_t|] = \Ocal\lp \frac{1}{s_N}\expec[X_t^2]\rp=\Ocal\lp \frac{1}{s_N}\Var(X_t)\rp =\Ocal\lp\frac{1}{s_N}\rp.
\end{align*}

We turn to $\expec[A_t]$. $K_t$ is non-random and $\expec[v_t]=\expec[w_t]=0$. Thus,
\begin{align*}
    \expec[A_t] = (1-K_t I(v))\lp\gamma\expec[X_{t-1}-\hat x_{t-1}]\rp.
\end{align*}
It follows from \cref{eq:bias-bound} that the bias $\expec[X_{t-1}-\hat x_{t-1}]$ is $\Ocal(1/s_N)$. It also follows from Lemma \ref{lem:KFPXg} that $1-K_tI(v)=\Ocal(1)$. Combining, it follows that $\expec[\remainder_t]\expec[A_t]=\Ocal(1/s_N^2)$. We underline that our arguments used Lemma \ref{lem:KFPXg}, which is proved independent of Lemma \ref{lem:remainderterms1}. Thus, the argument is not circular. 

 To bound $\Cov(\remainder_t,A_t)$, we now only need to bound $\expec[\remainder_tA_t]$.
\begin{align} 
\Ex[\remainder_tA_t]  = & \Ex\left[X_t(\phi'(v_t)-\Ex[\phi'(v_t)])(1-K_tI(v))\gamma(X_{t-1}-\hat{x}_{t-1})\right] + \tag{i}\\
&   \frac{1}{\sqrt{N}}\Ex\left[X_t(\phi'(v_t)-\Ex[\phi'(v_t)])(1-K_tI(v))w_t\right]+ \tag{ii}\\
& s_N\Ex\left[X_t(\phi'(v_t)-\Ex[\phi'(v_t)])(1-K_tI(v))K_t v_t\right]\pm \tag{iii}\\
& \frac{1}{ s_N}\|\phi''\|_{\infty}\Ex\left[X_t^2|1-K_tI(v)||(X_{t-1}-\hat{x}_{t-1})|\right] \pm \tag{iv}\\ 
& \frac{1}{ s_N}\|\phi''\|_{\infty}\Ex\left[X_t^2|1-K_tI(v)||w_t|\right] \pm \tag{v}\\
& \frac{1}{ s_N}\|\phi''\|_{\infty}\Ex\left[X_t^2K_t |v_t|\right].  \tag{vi}
\end{align}

Using the independence of $v_t$, and $\expec[\phi(v_t)-\expec\phi(v_t)]]=0$, it follows that (i) and (ii) are $0$. In (iii), $X_t$ is independent of $v_t$ and $\expec[X_t]=0$ in stationarity so that (iii) is, as well, equal to $0$. 

We turn to (iv)-(vi) starting with (iv). 

\begin{align*} 
\frac{1}{ s_N}\|\phi''\|_{\infty}\Ex\left[X_t^2|1-K_tI(v)||(X_{t-1}-\hat{x}_{t-1})|\right] 
& \leq \frac{1}{s_N}|1-K_tI(v)|\sqrt{\Ex[X_t^4]}\sqrt{\Var(X_{t-1}-\hat{x}_{t-1})}\\ 
& \leq \frac{1}{s_N}|1-K_tI(v)|\sqrt{\Ex[X_t^4]}\,(1+\Var(X_{t-1}-\hat{x}_{t-1}))\\
& =\Ocal\left(\frac{1}{s_N}(1+\Var(X_{t-1}-\hat{x}_{t-1}))\right)\\
& =\Ocal\lp \frac{1}{s_N}\lp 1+\hat{P}_{t-1} \rp \rp.
\end{align*} 

Rows (v) and (vi) are argued similarly. By Lemma \ref{lem:KFPXg}, $K_t= \Ocal(1/(s_N\sqrt{N}))$ and 
we conclude that $$\left|K_t\Cov(\remainder_t,A_t)\right| =\Ocal\lp  \frac{1}{s_N^2\sqrt{N}}(1+\hat{P}_{t-1})\rp,$$ as required.  \vspace{0.5cm}
\end{proofof}

\begin{proofof}{Lemma \ref{lem:preliminary}}
   
Recall that 
$$\hat{x}_{t+1}^c = \gamma\hat{x}_{t}^c + K_{t}s_N\phi\lp \frac{Y_{t+1}-\gamma\hat{x}_{t}^c}{s_N}\rp.$$

As before, let $e_t:=X_t-\hat x_t$. Because $X_{t+1}= \gamma X_{t}+\frac{1}{\sqrt{N}}w_{t}$,   
$$X_{t+1}-\hat{x}_{t+1}^c = \gamma(X_{t}-\hat{x}_{t}^c)+\frac{1}{\sqrt{N}}w_{t} - K_{t}s_N\phi\lp v_{t+1}+\frac{\gamma (X_{t}-\hat{x}_{t}^c)+\frac{1}{\sqrt{N}w_{t}}}{s_N} \rp$$
and with $U_{t}= \gamma e_t+\frac{1}{\sqrt{N}}w_t$ we have, $$e_{t+1} = U_{t} -K_{t} s_N \phi \lp v_{t+1}  +\frac{U_{t}}{s_N}\rp.$$

Fixing an even integer $l\geq 2$ we have, by expanding the power series of $e_t$, 
$$e_{t+1}^l = U_t^l - l U_t^{l-1} K_{t} s_N \phi \lp v_{t+1}  +\frac{U_{t}}{s_N}\rp + \sum_{m\geq 2: l-m\geq 0}C_m U_t^{l-m}K_{t}^m s_N^m \phi^m\lp v_{t+1}  +\frac{U_{t}}{s_N}\rp,$$  where $C_m=(-1)^m{l\choose m}$ are the coefficients of the binomial expansion $(1-x)^l$. 

The terms inside and outside the summation will be handled separately. First, observe from the tower property that
\begin{align*}
    \Ex\left[U_{t}^{l-1}\phi \lp v_{t+1}  +\frac{U_{t}}{s_N}\rp\right]=\expec\lb \Ex\left[U_{t}^{l-1}\phi \lp v_{t+1}  +\frac{U_{t}}{s_N}\rp\Big|U_t\right] \rb.
\end{align*}

By strong dissipativity in Assumption \ref{asum:centered}
$$\Ex\left[U_{t}^{l-1}\phi \lp v_{t+1}  +\frac{U_{t}}{s_N}\rp \Big| U_t\right]=s_N U_{t}^{l-2}\Ex\left[\frac{U_t}{s_N}\phi \lp v_{t+1}  +\frac{U_{t}}{s_N}\rp \Big| U_t\right]\geq s_N \zeta \frac{U_t^{l}}{s_N^2}.$$

In turn, using that $1-\eta K_t\geq 0$ for all $N$ sufficiently large for have some $\eta$ fixed
 $$\left|U_t^l - \Ex\left[lU_{t}^{l-1}K_ts_N \phi \lp v_{t+1}  +\frac{U_{t}}{s_N}\rp U_t\right]\right| \leq (1-\eta K_t)|U_t^l|.$$

For $m\geq 2$ we have, using the growth bound in Assumption \ref{asum:centered}, that for a constant $\Gamma_m$ which depends only on $\expec[v^m]$
$$\left|\Ex\left[\phi^m\lp v_{t+1}  +\frac{U_{t}}{s_N}\rp\Big|U_t\right]\right|\leq \Gamma_m^{intial}\lp 1+\lp \frac{|U_t|}{s_N}\rp^m \rp,$$ so that, recalling $K_t = \Ocal(1/(s_N\sqrt{N}))$, we have, for some different $\Gamma_m$
$$|U_t|^{l-m} K_{t}^m s_N^m \left|\Ex\left[\phi^m\lp v_{t+1}  +\frac{U_{t}}{s_N}\rp\Big|U_t\right]\right| \leq \Gamma_m \lp \frac{1}{\sqrt{N}^m} + \frac{|U_t|^{l-m}}{\sqrt{N}^m} + \lp \frac{1}{s_N\sqrt{N}} \rp^m |U_t|^l\rp $$

Notice that $K_t^m\leq \frac{1}{2}K_t$ for all $N$ large enough. In particular, there exists $N_0$ large enough so that for all $N\geq N_0$, and for some constant $\eta\in (0,1)$ 

Then for all $l\geq 2$ even $$\Ex[|e_{t+1}|^l|U_t] \leq (1-\eta K_t)|U_t^l| +\frac{1}{\eta}\sum_{m\geq 2:l-m\geq 0}\frac{1+|U_t|^{l-m}}{\sqrt{N}^m}.
$$ 

Recall now that $U_t=\gamma e_t+\frac{1}{\sqrt{N}}w_t$, so that---for $l$ even and a possibly re-defined $\eta$  
$$\Ex[e_{t+1}^l]\leq (1-\eta K_t)|\Ex[e_t^l] + \frac{1}{\eta\sqrt{N}^l} + \frac{1}{\eta}\sum_{m\geq 2:l-m\geq 0}\frac{1+\Ex[|e_t|^{l-m}]}{\sqrt{N}^m}.$$

\textbf{Case $l=2$:} Now, let us assign specific values to $l$. For $l=2$ we have $\eta>0$ and a recursion 
$$\Ex[e_{t+1}^2]\leq (1-\eta K_t)\Ex[e_t^2] + \frac{1}{\eta N}.$$ Iterating the recursion, and recalling $K_t= 1/(s_N\sqrt{N})$ we have in stationarity that 
$$P_{t+1}^c := \Ex[e_{t+1}^2] = \Ocal\lp \frac{1/N}{1/(s_N\sqrt{N})}\rp = \Ocal\lp \frac{s_N}{\sqrt{N}}\rp,$$ as stated. 

Using Cauchy-Schwarz inequality, this, in particular, implies that $\Ex[|e_t|]\leq \sqrt{P_t^c} =\Ocal(1)$ for all $s_N\lesssim \sqrt{N}$. 

\textbf{Case $l=4$:} For $l=4$ we have, for a suitable $\eta$, 
$$\Ex[e_{t+1}^4]\leq (1-\eta K_t)\Ex[e_t^4] + \frac{1}{\eta N^4} + \frac{1}{\eta} \left(\frac{\Ex[e_t^2]}{N} + \frac{\Ex[|e_t|]}{\sqrt{N}^3}\right).$$ Thus, we have 
$$\Ex[e_{t+1}^4]\leq (1-\eta K_t)\Ex[e_t^4] + \Ocal\lp \frac{s_N}{N^{3/2}}\rp.$$ Iterating the recursion we get that, in stationarity, 
$$F_t^c = \Ex[e_t^4] = \Ocal\lp \frac{s_N/N^{3/2}}{1/(s_N\sqrt{N})}\rp=\Ocal\lp\frac{s_N^2}{N} \rp=\Ocal((P_t^c)^2).$$

Finally, by Jensen's inequality with the concave function $g(x) = x^{3/4}$,  we have $$\Ex[|e_t|^3] \leq (\Ex[e_t^4])^{3/4} = \Ocal \lp\frac{s_N^{3/2}}{N^{3/4}}\rp$$ as stated. 

 \end{proofof} 

\begin{proofof}{Lemma \ref{lem:remainderterms2}}

Recall that \begin{align*} \varrho_t&:= s_N\phi\lp\frac{Y_t-\gamma\hat{x}_{t-1}^c}{s_N}\rp - s_N\phi(v_t)-\Ex[\phi'(v_t)](X_t-\gamma\hat{x}_{t-1}^c)\\ & ~= 
(X_t-\gamma \hat{x}_{t-1}^c)(\phi'(v_t)- \Ex[\phi'(v_t)])+ \frac{1}{2} \phi''(\varphi_t)\frac{1}{s_N}(X_t-\gamma \hat{x}_{t-1}^c)^2.\end{align*} Because $X_t,\hat{x}_{t-1}^c$ are independent of $v_t$, we have that 
\begin{equation} \label{eq:varrhobound_centered} 
|\Ex[\varrho_t]|\leq \frac{1}{s_N} \|\phi''\|_{\infty}\Ex[(X_t-\gamma\hat{x}_{t-1}^c)^2]\leq \frac{2}{s_N} \|\phi''\|_{\infty}\lp \Ex[(X_{t-1}-\hat{x}_{t-1}^c)^2] + \frac{1}{N}\rp =\Ocal\lp  \frac{1}{\sqrt{N}} \rp, 
\end{equation} where the last equality follows from Lemma \ref{lem:preliminary}. In particular, $(\Ex[\varrho_t])^2 = \Ocal(1/N)$. 

It similarly follows from Lemma \ref{lem:preliminary}, using the fourth moment bound that 
\begin{align} 
\Ex[(\varrho_t)^2]= \Ocal(s_N/\sqrt{N}+s_N^2/N) = \Ocal(s_N/\sqrt{N}),  
\end{align} where the last equality holds because $s_N^2/N\lesssim s_N/\sqrt{N}$ for $s_N\lesssim \sqrt{N}.$

Using Lemma \ref{lem:KFPXg}, specifically that $K_t = \Ocal(1/(s_N\sqrt{N}))$, we then have that
$$Var(K_t\varrho_t) = K_t^2 Var(\varrho_t) =  \Ocal\lp \frac{1}{s_N N^{\frac{3}{2}}}\rp.$$

We turn to $Cov(K_t\varrho_t,A_t)$. 
Recall that 
$$A_t:=  
(1-K_tI(v))(X_t-\gamma \hx_{t-1}^c)) - K_ts_Nv_t.$$
Because $v_t,w_{t-1}$ are independent of each other and of $X_{t-1}-\hat{x}_{t-1}^c$, and because $K_t\ll 1$ for $N$ large enough, we have that 
$$|\Ex[A_t]| \leq |\Ex[X_{t-1}-\hat{x}_{t-1}^c]| =\Ocal(1/\sqrt{N}),$$ where the equality follows from the bound on the bias \eqref{eq:biascentered}. The latter does not use Lemma \ref{lem:remainderterms2} so that there is no circularity. 

Using Lemma \ref{lem:KFPXg} and \eqref{eq:varrhobound_centered} we have 
$$K_t\Ex[\varrho_t]\Ex[A_t] =\Ocal\lp \frac{1}{s_N\sqrt{N}}\frac{1}{\sqrt{N}}\frac{1}{\sqrt{N}}\rp = \Ocal\lp \frac{1}{s_N N^{\frac{3}{2}}}\rp.$$

\begin{align} A_t\varrho_t =& (1-K_tI(v)) (X_{t}-\gamma \hat{x}_{t-1}^c)^2(\phi'(v_t)-\Ex[\phi'(v_t)])\pm  \tag{I} \\ 
&  \tag{II} 
(1-K_tI(v))\frac{1}{2}\|\phi''\|_{\infty}\frac{1}{s_N} (X_{t}-\gamma \hx_{t-1}^c)^3+\\ \tag{III}& 
(-K_ts_Nv_t)(X_t-\gamma\hat{x}_{t-1}^c)(\phi'(v_t)-\Ex[\phi'(v_t)])\pm 
\\  & (-K_tv_t)\frac{1}{2}\|\phi''\|_{\infty} (X_{t}-\gamma \hx_{t-1}^c)^2\tag{IV}
\end{align} The first row has trivially  $\Ex[\text{I}] = 0$ because $v_t$ is independent of $X_t-\gamma\hat{x}_{t-1}^c$. 
Next, using Lemma \ref{lem:preliminary}
$$|\Ex[\text{II}]|\leq \|\phi''\|_{\infty}\frac{1}{s_N} |\Ex[(X_t-\gamma\hat{x}_{t-1})^3]=\Ocal\lp\frac{\sqrt{s_N}}{N^{3/4}}\rp.$$
Here we used that 
$\Ex[(X_t-\gamma\hat{x}_{t-1}^c)^3]\leq 
4 |\Ex[(X_{t-1}-\hat{x}_{t-1}^c)^3]
| +  \frac{4\Ex[|w_1|^3]}{N^{3/2}}$.

Using the independence of $v_t,w_{t-1}$ and $(X_{t-1}-\hat{x}_{t-1}^c)$ we have that 
$$|\Ex[\text{III}]| = \gamma K_ts_N |\Ex[v_t(\phi'(v_t)-\Ex[\phi'(v_t)])]|| \Ex[X_{t-1}-\hat{x}_{t-1}^c]|=\Ocal\lp \frac{1}{N}\rp.$$ Here we used $K_t= \Ocal(1/(s_N\sqrt{N})$, the bias bound $\Ex[X_{t-1}-\hat{x}_{t-1}^c]=\Ocal(1/\sqrt{N})$ (see \eqref{eq:biascentered}) 
and $|\Ex[v_t(\phi'(v_t)-\Ex[\phi'(v_t)])]|\leq \sqrt{\Ex[v_t^2]Var(\phi'(v_t))} =\Ocal(1)$.

The last row has by Lemma \ref{lem:preliminary} that $$\Ex[|\text{IV}|] \leq 2 K_t  \|\phi''\|_{\infty}\lp P_t^c + \frac{1}{N}\rp = \Ocal\lp \frac{1}{N}\rp.$$

Finally, multiplying by $K_t= \Ocal(1/(s_N\sqrt{N})) $we have that 
$$K_t\Ex[\varrho_tA_t] = \Ocal\lp \frac{1}{s_N\sqrt{N}}\frac{\sqrt{s_N}}{N^{3/4}}\rp = \Ocal\lp \frac{1}{\sqrt{s_N}N^{\frac{5}{4}}}\rp,$$ so that we may conclude that 
$$|Cov(K_t\varrho_t,A_t)|= \Ocal\lp \frac{1}{\sqrt{s_N}N^{\frac{5}{4}}}\rp.$$ Because $s_NN^{3/2} \geq \sqrt{s_N}N^{5/4}$ for $s_N=\Omega(1)$ we finally conclude that 
$$K_t^2Var(\varrho_t^2) +|Cov(K_t\varrho_t,A_t)|=\Ocal\lp \frac{1}{\sqrt{s_N}N^{\frac{5}{4}}}\rp, $$ as stated. 
\end{proofof}

\begin{proofof}{Lemma \ref{lem:lecamapp}}

Let us assume without loss of generality that $\theta=0$ and let $f_{\tau}(z):=f_{\zeta}(z-\tau)$ denote the density of $Z$ (the observation) when $\theta=\tau$; by assumption $Z$ has support on the whole real line. Using Bayes rule
\begin{align*}
   \mathbb{P}[\tau(Z)=\tau]=\frac{f_{\tau}(z)\pi(\tau)}{\int_{\tau}f_{\tau}(z)\pi(\tau)d\tau}.\numberthis \label{eq:bayes} 
\end{align*}
Under the hypothesis of the lemma, there exists $L_\pi$ such that 
\begin{align*}
    \sup_{\tau_1,\tau_2}\frac{|\pi(\tau_2)-\pi(\tau_1)|}{|\tau_2-\tau_1|}\leq L_\pi.
\end{align*}
in particular, for any $\tau$, $|\pi(\tau)-\pi(0)| \leq |\tau |L_\pi$. Then, $$f_{\tau}(z)\pi(\tau)= f_{\tau}(z)\pi(0) \pm L_{\pi} |\tau| f_{\tau}(z),$$
and 
$$\int_{\tau} f_{\tau}(z)\pi(\tau)d\tau= \pi(0)\int_{\tau}f_{\tau}(z)d\tau \pm L_{\pi} \int_{\tau}|\tau| f_{\tau}(z)d\tau.$$
To bound the second term, write 
$$\int_{\tau}|\tau| f_{\tau}(z)d\tau =
\int_{y}|z-y| f_{\zeta}(y)dy\leq z + \int_y |y|f_{\zeta}(y)dy,$$
By the Cauchy-Schwarz inequality we have $\int_y |y|f_{\zeta}(y)dy\leq \sigma=\ocal(1)$. Therefore, fixing $b=\ocal(1)$ and using $\int_{\tau}f_{\tau}(z)d\tau = \int_{\tau}f_{\zeta}(z-\tau)d\tau = 1$, we have $z\in [-b,b]$ that 
\[
\int_{\tau} f_{\tau}(z)\pi(\tau)d\tau = \pi(0) + \ocal(1).
\]

In turn, recalling that $\pi(0)>0$, we have for $z\in [-b,b]$ that 
\begin{align*}
\frac{f_{\tau}(z)\pi(\tau)}{\int_{\tau}f_{\tau}(z)\pi(\tau)d\tau}   & = 
\frac{f_{\tau}(z)}{\int_{\tau}f_{\tau}(z)d\tau+\ocal(1)}\pm \frac{L_{\pi}|\tau|f_{\tau}(z)}{\int_{\tau}f_{\tau}(z)d\tau+\ocal(1)}\\
& =\frac{f_{\tau}(z)}{\int_{\tau}f_{\tau}(z)d\tau}(1+\ocal(1))\pm \frac{L_{\pi}}{\pi(0)}|\tau|\frac{f_{\tau}(z)}{\int_{\tau}f_{\tau}(z)d\tau}(1+\ocal(1)) \\
& =f_{\zeta}(z-\tau)(1+\ocal(1))\pm  2\frac{L_{\pi}}{\pi(0)}|\tau|f_{\zeta}(z-\tau),
\end{align*} where in the second equality we use $\int_{\tau}f_{\tau}(z)d\tau =1$, 
Overall for $z\in[-b,b]$, 
integrating over $\tau$ we have for $z\in [-b,b]$  \begin{align*} \Ex[\tau(Z)|Z=z] & = \int_{\tau}\tau f_{\zeta}(z-\tau)(1+\ocal(1))d\tau \pm 2\frac{L_{\pi}}{\pi(0)}\int_{\tau}|\tau|^2f_{\zeta}(z-\tau)d\tau  \\ & = z(1+\ocal(1))+\Ocal(z^2)=z(1+\ocal(1)),\end{align*} In the first equality we use $\int_{\tau}\tau f_{\zeta}(z-\tau)d\tau = \int_y (z-y)f_{\zeta}(y)dy = z$ because $\Ex[\zeta]=0.$. The last equality follows recalling that $b=\ocal(1)$ so that $z^2 = o(z)$. 

Finally, because $\zeta = \sigma\zeta_0$ (with $\zeta_0$ having a fourth finite moment), we can choose $b=\ocal(1)\gg \sigma$ so that $\Ex[Z^21\{|Z|>b\}] =\ocal(\sigma^2).$ Recalling that $\theta=0$ and that theconditional expectation minimizes the MSE, 
\begin{align*} \Ex[(\hat{\theta}-\theta)^2]& \geq \Ex[(\tau(Z))^2]\geq \int_{z\in [-b,b]}z^2 f_{\zeta}(z)(1+\ocal(1)) dz  = \sigma^2(1+\ocal(1)),
\end{align*} as stated. 
\end{proofof}
\ifCLASSOPTIONcaptionsoff
  \newpage
\fi

\sloppy
\bibliography{biblio,itaibib}
\fussy

\end{document}